\documentclass[11pt]{article}
\usepackage{amsthm}
\usepackage{amsmath}
\usepackage{dsfont}
\usepackage{amssymb}
\usepackage{graphicx}
\usepackage{hyperref}
\usepackage{tikz}
\usepackage{cases}
\usepackage{empheq}
\usepackage{enumerate}
\usepackage{xcolor}
\definecolor{darkgreen}{rgb}{0.1,0.5,0.1}
\hypersetup{colorlinks,linkcolor={red},citecolor={darkgreen},urlcolor={red}} 
\usetikzlibrary{arrows,positioning,fit,shapes,calc}

\newtheorem{Thm}{Theorem}
\newtheorem{Lem}[Thm]{Lemma}
\newtheorem{Def}[Thm]{Definition}

\newtheorem{Cor}[Thm]{Corollary}
\newtheorem{Rem}[Thm]{Remark}
\numberwithin{equation}{section}
\newtheorem{Example}[Thm]{Example}
\newtheorem{Ass}[Thm]{Assumption}

\newcommand{\beqa}{\begin{eqnarray}}
\newcommand{\eeqa}[1]{\label{#1}\end{eqnarray}}
\newcommand{\beq}{\begin{equation}}
\newcommand{\eeq}[1]{\label{#1}\end{equation}}

\newcommand{\rmd}{{\mathrm{ d}}}
\newcommand{\rme}{{\mathrm{ e}}}
\newcommand{\rmi}{{\mathrm{ i}}}

\newcommand{\R}{{\mathbb{ R}}}
\newcommand{\C}{{\mathbb{C}}}
\newcommand{\bbR}{{\mathbb{R}}}
\newcommand{\bbC}{{\mathbb{C}}}
\newcommand{\bbA}{{\mathbb{A}}}
\newcommand{\N}{{\mathbb{N}}}
\newcommand{\md}{{\mathrm{ d}}}

\newcommand{\bal}{{\boldsymbol{\alpha}}}

\newcommand{\eps}{\varepsilon}

\newcommand{\bk}{\mathbf{k}}

\newcommand{\bx}{\mathbf{x}}

\newcommand{\bu}{\mathbf{u}}

\newcommand{\bE}{\mathbf{E}}
\newcommand{\bH}{\mathbf{H}}

\newcommand{\bL}{\mathbf{L}}

\newcommand{\bU}{\mathbf{U}}

\newcommand{\bbP}{\mathbb{P}}
\newcommand{\bbM}{\mathbb{M}}

\newcommand{\bbX}{\mathbb{X}}

\newcommand{\calP}{{\mathcal{P}}}

\newcommand{\calZ}{{\mathcal{Z}}}

\newcommand{\ds}{\displaystyle}

\def\bR{\boldsymbol{\rm{R}}}

\def\XXint#1#2#3{{\setbox0=\hbox{$#1{#2#3}{\int}$}
		\vcenter{\hbox{$#2#3$}}\kern-.5\wd0}}

\begin{document}
\vspace{-1in}

\title{An operator approach to the analysis of electromagnetic wave propagation in dispersive media. Part 1: general results.}

\author{Maxence Cassier$^{a}$ and  Patrick Joly$^{b}$ \\ \\
{
\footnotesize $^a$ Aix Marseille Univ, CNRS,  Centrale Med, Institut Fresnel, Marseille, France}\\ 
{\footnotesize  $^b$ POEMS$^1$, CNRS, INRIA, ENSTA Paris, Institut Polytechnique de Paris, 91120 Palaiseau, France}\\ 
{\footnotesize (maxence.cassier@fresnel.fr, patrick.joly@inria.fr)}}
\maketitle
	\begin{abstract}
We investigate in this chapter the mathematical models for electromagnetic wave propagation
in dispersive isotropic passive linear media  for which the dielectric permittivity $\varepsilon$ and magnetic permeability $\mu$ depend on the frequency. 
We emphasize the link between physical requirements  and
mathematical properties of the models. A particular attention is devoted to the notions
of causality and  passivity  and its connection to  the existence of Herglotz functions that determine the dispersion of the material. We consider successively the cases of the general passive media and the so-called local media for which $\varepsilon$ and 
$\mu$ are rational functions of the frequency. This leads us to analyse  the important class of non dissipative and  dissipative generalized Lorentz models. In particular, we discuss the connection between mathematical and physical properties of models through the notions of stability, energy conservation, dispersion and modal analyses, group and phase velocities and energy decay in dissipative systems.
	\end{abstract}	


\section{Introduction}
The phenomenon of dispersion of waves is represented by the fact that their speed of propagation depends of their wavelength. Maxwell's equations in the vacuum are non dispersive. However, according to physicists \cite{Jac-98, landau-84, Tip}, dispersion is ubiquitous  for electromagnetic wave propagation in matter. This property is traduced in the models via the dependence of the material properties (electric permittivity $\varepsilon$ and magnetic permeability $\mu$) of the propagation media with respect to the frequency $\omega$. \\[12pt]
The theory of wave propagation in dispersive media, and more specifically negative index materials in electromagnetism, 
had known recently a regain of interest with the appearance of electromagnetic metamaterials. Their  theoretical behaviour had been, much before their experimental realization, predicted in the pioneering article of Veselago \cite{veselago}. Since 
the beginning of the century, several works \cite{smith2004metamaterials}, \cite{cui2010metamaterials}, \cite{brien2002photonic} have shown a practical realisability of metamaterials, with the help of a periodic assembly of small resonators whose effective macroscopic behaviour corresponds to
a negative index. Their existence opened new perspectives of application for physicists, in particular in optics and photonic crystals,  related to new physical phenomena such as backward propagating waves, negative refraction \cite{veselago} or plasmonic surface waves \cite{Mai-07}, which are used for creating perfect lenses \cite{pendry2000negative},  superlensing \cite{Mil-05} or cloaking \cite{Nir-94,Mil-06}. On the other hand the study of the corresponding 
mathematical models raised new exciting  questions for mathematicians (see \cite{JLiReview,cas-kach-jol-17} for a review), in particular numerical analysts \cite{JLiBook,ziolkowski2001wave,Bec-15,JLiLorentz}.\\[12pt]
\noindent The organisation of the chapter is as follows. Sections \ref{sec-presentation} to \ref{sec-Gen-Lorentz} are based on the article \cite{cas-kach-jol-17}. After a brief recap on the Fourier--Laplace transform in section \ref{sec-Laplace}, we describe in section  \ref{sec-presentation} the relevant constitutive laws for dispersive electromagnetic materials, that are obtained by imposing some restrictions  to the  dependence of $\varepsilon$ and $\mu$ as functions of $\omega$, linked to natural physical requirements. In section \ref{sec-passive}, we specify ourselves to the class of passive materials which  naturally  leads to the notion of Herglotz functions. Section \ref{Generalpassive} is devoted to Maxwell's equations in such passive media for which we provide a so-called conservative augmented formulation in section \ref{Augmented}.  This  formulation is based on the well-known Herglotz-Nevanlinna representation of Herglotz functions that is presented in section \ref{Representationforula}. In section \ref{sec-abstract}, these equations are reformulated, in the absence source term, in an abstract manner as 
\begin{equation}\label{eq.schroding}
\frac{\md \, {\bf U}}{\md\, t} + \rmi\,\mathbb{A} \, {\bf U}=0,
\end{equation}
where $\bbA$ is a self-adjoint in some appropriate Hilbert space $\mathcal{H}$,
which allows  to state a uniqueness and existence result for the related initial-value  problem. Section \ref{sec-Gen-Lorentz} treats an important subclass of passive media: the generalized Lorentz media. The corresponding models are presented in section \ref{sec-Lorentz-def} and reformulated in the abstract form \eqref{eq.schroding} in section \ref{sec-abstract2}. These models  are non-dissipative  and support the propagation of non-dissipative plane waves in homogeneous media.  In section \ref{sec-disp-analysis}, we provide a complete  dispersion analysis of such plane waves. This leads us to introduce  the important notions of forward modes, backward  modes and negative index materials. The last section \ref{Lorentz-dissp} is devoted to the study of a dissipative version of these generalized  Lorentz models introduced in section \ref{Lorentz-dissp-permittivity-permeability} and  whose evolution system is given  in section  \ref{sec-disspative-evol} via a dissipative augmented formulation  rewritten in an abstract way in section \ref{sec_decay_proof}. This formulation is of the same form as \eqref{eq.schroding} with  $\bbA$ replaced by $\bbA_{\bal}$ where $\bal$ is a set of positive dissipation parameters and the operator  $\bbA_{\bal}$ is no longer self-adjoint.
The rest of the section the analysis of the large time-behavior of the solutions and follows the content of the two articles \cite{cas-jol-ros-22,cas-jol-ros-22-bis}. Our main results stated in section \ref{sec_decay} are based on the analysis of the dispersion relation for these dissipative models whose study is presented in section \ref{sec_DispersionDissipative}. Finally in section  \ref{sec_decay_proof}, we sketch the proof  that we adopted in the second article \cite{cas-jol-ros-22-bis}, proof which is  based on the spectral decomposition of the operator $\bbA_{\bal}$.

\section{Recap on the Fourier-Laplace transform in time} \label{sec-Laplace}
Let $u(t)$ be a (measurable) complex-valued, locally bounded and causal ($u(t) = 0 \mbox{ for } t < 0 \}$) function of time, which
we suppose to be {exponentially bounded} for large $t$ (for simplicity). More precisely, given $\alpha \in \R$, we introduce
\begin{equation} \label{defPB}
	PB_\alpha (\R^+) = \{ u: \R^+ \rightarrow \C / \; \exists \; (C, p) \in \R_{+} \times \N \  \mid |u(t)| \leq C \, e^{\alpha t} \, (1 + t^p),\, \forall \; t\in \mathbb{R}^+ \}.
\end{equation} 
(where $\R^+:=\{ t\in \R\mid t\geq 0\}$  and $\R_+:=\R^+\setminus \{ 0\}$).
For $\alpha = 0$, $PB_0 (\R^+)$ is the class of {polynomially bounded} functions. The Fourier-Laplace transform $\widehat{u}(\omega)$ of $u$ is  defined in the complex half space (see e.g. \cite{Say-16}, proposition 3.13):
\begin{equation} \label{defC+}
	\C_\alpha^+ = \{ \omega \in \C \;  / \; \operatorname{Im} \, \omega > \alpha \}, \quad (\mbox{where } \C_0^+ \mbox{ will be denoted by } \C^+ \mbox{ when }  \alpha = 0)
\end{equation}
by the following integral formula (we use here the convention which is usual for physicists)
\begin{equation} \label{defFLT}
	\forall \; \omega \in \C^+, \quad  \widehat{u}(\omega) = \frac{1}{\sqrt{2\pi}} \int_0^{+\infty}
	u(t) \; e^{\rmi \omega t} \; dt.
\end{equation}
Note that, with this convention, as soon as $u$ and $\partial_t u$ belong to $PB_\alpha (\R^+)$, we have
\begin{equation} \label{propFLT}
 \forall \; \omega \in \C_\alpha^+, \quad  \mbox{with } v := \partial_t u, \quad \widehat{v}(\omega) = - \rmi \omega \, \widehat{u}(\omega) + u(0), 
\end{equation}
which reduces to $\widehat{v}(\omega) = - \rmi \omega \, \widehat{u}(\omega)$ when $u(0)=0$.\\[12pt]
This transform is related to the usual Fourier transform $u(t) \rightarrow {\cal F}u(\omega)$ by 
\begin{equation} \label{Link_F-FLT}
	\forall \; \eta > \alpha, \quad \forall \; \omega \in \R, \quad \widehat{u}(\omega + \rmi\eta) = {\cal F}\big(u \, e^{-\eta t})(\omega)
\end{equation}
which proves in particular that  $\omega \in \R \mapsto \widehat{u}(\omega + \rmi\eta) \in L^2(\R)$ for any $\eta > \alpha$ and that 
\begin{equation} \label{FLT_L2}
 \int_{-\infty}^{+\infty} |\widehat u(\omega + \rmi\eta)|^2 \; d\omega = \int_0^{+\infty} |u(t)|^2 \, e^{-2\eta t} \; dt,
\end{equation}
by Plancherel's theorem. Moreover, one easily sees that
\begin{equation} \label{Analyticity}
	\forall \; u \in PB_\alpha(\R^+),  \quad \omega \mapsto \widehat{u}(\omega) \mbox{ is analytic in } \C_\alpha^+.
\end{equation}
One can expect that $\widehat{u}(\omega)$ can be extended as an analytic function in a domain of the complex plane that contains the half-space $\C_\alpha^+$. When needed, we shall use the same notation $\widehat{u}(\omega)$ for this extension. In the following $\omega$ will be referred to as the (possibly complex) {frequency}.\\[12pt]
The half-plane $ \C_\alpha^+$ in invariant under the transformation $\omega \rightarrow - \, \overline{\omega}$, which corresponds to the symmetry with respect to the imaginary axis. Laplace-Fourier transforms of {\bf real-valued} functions have a particular property with respect 
to this transformation in $PB_\alpha (\R^+)$
\begin{equation} \label{FLTrealfields}
	u(t) \in \R, \quad \forall \; t\geq 0, \quad \Longleftrightarrow \quad	\forall \; \omega \in \C^+, \quad  \widehat{u}(- \, \overline{\omega}) = \overline{\widehat{u}(\omega)}
\end{equation}
In the sequel, we assume that all the functions of time that are used in this article (for instance, any component of the electromagnetic field at a given point), belong to some $PB_\alpha(\R+)$.
\section{Maxwell's equations in dispersive media:  general properties} \label{sec-presentation}
Maxwell's equations relate the space variations of the electric and magnetic fields
${\bf E}({\bf x},t)$ and ${\bf H}({\bf x},t)$ (where ${\bf x} \in \R^3$ denotes the space variable and $t>0$ is the time) to the time variations of the corresponding electric and magnetic inductions ${\bf D}({\bf x},t)$ and ${\bf B}({\bf x},t)$:
\begin{equation} \label{Maxwell}
	\partial_t {\bf B} + {\bf rot} \, {\bf E} = 0, \quad \partial_t {\bf D} - {\bf rot} \, {\bf H} = 0, \quad {\bf x} \in \R^3, \quad t >0.
\end{equation}
These equations need to be completed by so-called {constitutive laws} that characterize the material in which electromagnetic waves propagate by relating the electric (or magnetic) field and the corresponding induction. In this paper, we shall restrict ourselves to materials which are linear, {local in space} and time-independent.\\[12pt]
\noindent In standard isotropic {dielectric media}, the fields and induction are proportional
\begin{equation} \label{Dielectric}
	{\bf D}({\bf x},t) = \varepsilon({\bf x}) \, {\bf E}({\bf x},t) , \quad  {\bf B}({\bf x},t) = \mu({\bf x}) \, {\bf H}({\bf x},t),
\end{equation}
where   $\varepsilon({\bf x}) > 0$ and $\mu({\bf x}) > 0$ are  respectively the electric permittivity and the magnetic permeability of the material at the point ${\bf x}$. In the vacuum, these coefficients are of course independent of ${\bf x}$:
$
\varepsilon({\bf x}) = \varepsilon_0 = (36\pi)^{-1} \, 10^{-9} , \mu({\bf x}) = \mu_0 = 4\pi \, 10^{-7} .
$ \\ [12pt]
In the matter, the law (\ref{Dielectric}) must be seen only as an approximation (good as soon as $\varepsilon$ and $\mu$ are ``almost'' real and constant over a broad range of frequencies) of the reality because it would violate the high frequency principle, see ({\bf HF}) below. To account these phenomena, one needs to abandon the idea that the constitutive laws are local in time and to accept e.g. that ${\bf D}({\bf x},t)$ depends
on the history of  ${\bf E}({\bf x},t)$ between $0$ and $t$, i.e. \begin{equation} \label{causal_law}{\bf D}({\bf x},t) = F\big({\bf x}
	,t \, ; \big\{{\bf E}({\bf x},\tau), 0 \leq \tau \leq t \big\} \big).\end{equation}
The above obeys a fundamental physical principle: the causality principle. Adding the {time invariance principle} by translation, i.e. that the material behaves the same way whatever the time one observes it, one infers that the function $F$ is also independent of time:
$F(x,t \, ;\cdot) = F(x \, ;\cdot)$. This implies that these laws are given by time convolution products (in the sense of distributions) with kernels supported in $\R^+$ (see \cite{Zemanian}, sections 5.10  and 5.11).
\noindent In that case, it is useful to write constitutive laws in the {frequency domain} with the {Fourier-Laplace} transforms  of the fields and inductions as follows
\begin{equation} \label{Dispersive}
	\widehat{\bf D}({\bf x}, \omega) = \varepsilon({\bf x}, \omega) \, \widehat{\bf E}({\bf x},\omega) , \quad  \widehat{\bf B}({\bf x},\omega) = \mu({\bf x}, \omega) \, \widehat{\bf H}({\bf x},\omega).
\end{equation}
where for each ${\bf x}$, $\omega \in \C_\alpha^+  \mapsto \varepsilon({\bf x}, \omega)$ 
and $\omega \in \C_\alpha^+  \mapsto \mu({\bf x}, \omega)$ are  complex valued functions of the frequency (for some $\alpha \geq 0$). Of course, these functions satisfy some particular properties: \\ [12pt]
\noindent {\bf Causality principle:} if  ${\bf E}(\bx,t)$ (or ${\bf H}(\bx,t)$) is causal ${\bf D}(\bx,t)$ (or ${\bf B}(\bx,t)$) is causal too. According to the property \eqref{Analyticity}, this leads to impose \\ [12pt]
$
\hspace*{1.5cm} {\bf (CP)} \qquad \omega \mapsto \varepsilon({\bf x}, \omega) \quad \mbox{and} \quad \omega \mapsto \mu({\bf x}, \omega) \quad \mbox{are analytic in } \C_\alpha^+.
$ \\ [12pt]
{\bf Reality principle:}  if ${\bf D}(\bx,t)$ (or ${\bf B}(\bx,t)$) is real valued, ${\bf E}(\bx,t)$ (or ${\bf H}(\bx,t)$)
is real valued to too.  According to the property (\ref{FLTrealfields}), this leads to impose (\ref{Dispersive}) and   \\ [12pt]
$
\hspace*{1.5cm}{\bf (RP)} \qquad \forall \; \omega \in \C^+_{\alpha}, \quad  {\varepsilon}({\bf x}, - \, \overline{\omega}) = \overline{{\varepsilon}({\bf x},\omega)}, \quad  {\mu}({\bf x}, - \, \overline{\omega}) = \overline{{\mu}({\bf x},\omega)}.
$  \\ [12pt]
{\bf High frequency principle:} at high frequency, due to the inertia of charge carriers, any material `` behaves as the vacuum", see e.g. \cite{landau-84}, being less and less dispersive. Mathematically, this gives  \\ [12pt]
$
\hspace*{1.5cm}  \ds {\bf (HF)} \qquad \forall \; \eta > \alpha, \mbox{ if } \; \operatorname{Im}\, \omega \geq \eta >  \alpha, \quad \lim_{|\omega| \rightarrow + \infty} \varepsilon({\bf x}, \omega) =\varepsilon_0, \quad 
\lim_{|\omega| \rightarrow + \infty} \mu({\bf x}, \omega) =\mu_0.
$  \\ [12pt]
When ${\bf(CP})$ holds, the property ${\bf (HF)}$ implies that fields and inductions have the same {time regularity}, for instance
$
t \rightarrow {\bf E}({\bf x},t) \!  \in \! H^s_{loc}(\R^+) \, \Rightarrow \,  {\bf D}({\bf x},t)\!  \in \!H^s_{loc}(\R^+)
$, see \cite{cas-kach-jol-17}, section 2.1. It also plays a role in the well-posedness of  Maxwell's equations  (see e.g. remark 3.15 of \cite{cas-kach-jol-17}).
\\[12pt]
In the matter, for given inductions ${\bf D}(\bx,t)$ and ${\bf B}(\bx,t)$, one defines the electric polatisation  ${\bf P}(\bx,t)$ (resp. ${\bf M}(\bx,t)$) as the difference between the electric field ${\bf E}(\bx,t)$ (resp. the magnetic field ${\bf H}(\bx,t)$) with the value that it would have in the vacuum: 
\begin{equation} \label{defPM}
	 {\bf P}({\bf x}, t) := {\bf D}({\bf x}, t) - \varepsilon_0 \, {\bf E}({\bf x}, t), \quad {\bf M}({\bf x}, t) := {\bf B}({\bf x}, t) - \mu_0 \,  {\bf H}({\bf x}, t).	
\end{equation}
Then, Maxwell's equations can be rewritten as
\begin{equation} \label{Maxwell2}
	\varepsilon_0 \, \partial_t {\bf E} + {\bf rot} \, {\bf H} +  \partial_t {\bf P}= 0, \quad  \mu_0 \, \partial_t {\bf H} - {\bf rot} \, {\bf E}  + \partial_t {\bf M} = 0, \quad {\bf x} \in \R^3, \quad t >0,
\end{equation}
and the constitutive laws (\ref{Dispersive}) can  be rewritten as follows:
\begin{equation} \label{Dispersive2}
	\widehat{\bf P}({\bf x}, \omega) =  \varepsilon_0 \,{\chi}_e({\bf x}, \omega) \, \widehat{\bf E}({\bf x},\omega) , \quad  \widehat{\bf M}({\bf x},\omega) =  \mu_0 \, {\chi}_m({\bf x}, \omega) \, \widehat{\bf H}({\bf x},\omega).
\end{equation}
where the electric and magnetic susceptibilities $\widehat{\chi}_e({\bf x}, \omega)$ and $\widehat{\chi}_e({\bf x}, \omega)$ are defined by 
\begin{equation} \label{defsus}
\widehat{\chi}_e({\bf x}, \omega) = \frac{\varepsilon({\bf x}, \omega)}{\varepsilon_0} -1, \quad  \widehat{\chi}_m({\bf x}, \omega) = \frac{\mu({\bf x}, \omega)}{\mu_0} -1.
\end{equation}
\begin{Example} [$L^1_{loc}(\R^+)$-convolutional media] \label{Exconv} These correspond to the case where $\widehat{\chi}_e({\bf x}, \omega)$  and $\widehat{\chi}_m({\bf x}, \omega)$ are Fourier-Laplace transforms of causal real valued functions $t \mapsto \chi_e({\bf x}, t)$ and $t \mapsto \chi_m({\bf x},t)$ in $L^1_{loc}(\R^+)$. In that case, the constitutive laws rewrite, in the time domain,
	\begin{equation} \label{Dispersive_time}
		\left\{ 	\begin{array}{ll}
			\ds {\bf D}({\bf x}, t) = \varepsilon_0 \, \Big( {\bf E}({\bf x}, t) + \int_0^t \chi_e({\bf x}, \tau) \; {\bf E}({\bf x}, t-\tau) \; d \tau \Big),\\[12pt]
			\ds {\bf B}({\bf x}, t) = \mu_0 \, \Big( {\bf H}({\bf x}, t) + \int_0^t \chi_m({\bf x}, \tau) \; {\bf H}({\bf x}, t-\tau) \; d \tau \Big).
		\end{array} \right.
	\end{equation}
\end{Example}
\begin{Example} [Local media] \label{Exloc}
We shall say that a dispersive material is {\bf local} if and only if 
$\omega \mapsto  \varepsilon({\bf x}, \omega) \mbox{ and } \omega \mapsto \mu({\bf x}, \omega)  \mbox{ are (irreducible) rational fractions}$, implying the existence of polynomials in $z$, parametrized by ${\bf x}$, $z \mapsto P_e({\bf x}, z) , Q_e({\bf x}, z) , P_m({\bf x}, z) , Q_m({\bf x}, z)$ such that 
	$$	\ds	\varepsilon({\bf x}, \omega) =   \frac{P_e({\bf x}, -\rmi \omega)}{Q_e({\bf x}, -\rmi \omega)} , \quad \mu({\bf x}, \omega) =   \frac{P_m({\bf x}, -\rmi \omega)}{Q_m({\bf x}, -\rmi \omega)}.$$
The reader will  check that $({\bf CP}, {\bf RP}, {\bf HF})$ are satisfied  if only if these polynomials satisfy $d^oP_e = d^oQ_e$ and $d^oP_m = d^oQ_m$ and have real coefficients. The  (somewhat misleading) denomination of local media is due to the fact that 
(\ref{Dispersive}) can be rewritten in terms of  ODEs:
\begin{equation} \label{Dispersivelocal}
		Q_e({\bf x}, \partial_t) \; {\bf P} = \varepsilon_0 \, P_e({\bf x}, \partial_t) \; {\bf E}, \quad Q_m({\bf x}, \partial_t) \; {\bf M} =\mu_0 \, P_m({\bf x}, \partial_t) \; {\bf H}
\end{equation}
However, using the theory of linear ODE's, it is easy to see that these are a particular case of the media described in Example 1 where the kernels are linear combinations of products of exponential and trigonometric functions (see \cite{cas-jol-ros-22}, section 1.1.2 fore more details).
\end{Example}

\section{Passive materials} \label{sec-passive}

	\noindent  Defining the electromagnetic energy as in the vacuum, i.e.
	\begin{equation} \label{defEnergy}
		{\cal E}(t) :=  \frac{1}{2}\int_{\R^3} \big( \varepsilon_0 \, |{\bf E}|^2  + \mu_0 \, |{\bf H}|^2 \, \big)({\bf x}, t)\; d{\bf x}, \quad t >0,
	\end{equation}
	we shall say that a material is {\it physically} passive \cite{Gralak-10,cas-kach-jol-17,cas-jol-ros-22} if, when ${\bf E}$, ${\bf H}$, ${\bf D}$ and ${\bf B}$ are causal fields solving \eqref{Maxwell} in the absence of source term, the corresponding electromagnetic energy does not increase between $0$ and $T$ for any $T \geq 0$, namely,
	\begin{equation} \label{propEnergy}
		{\cal E}(T) \leq {\cal E}(0).
	\end{equation}
	This definition well corresponds to the intuitive notion hidden behind the word ``passive": a passive medium cannot create energy from its own.
	\begin{Rem} \label{decroissance} One needs to be careful on the fact that  property (\ref{propEnergy}) does not imply that ${\cal E}(t)$ is a decreasing function of time. Indeed, since (\ref{propEnergy}) is supposed to hold only for causal fields, the initial time $t=0$ cannot be replaced by any other initial tim $t_0$. See also the particular case treated in \cite{cas-kach-jol-17} and more precisely the Proposition 3.40. 
	\end{Rem}
\noindent Multiplying the first equation of (\ref{Maxwell2}) by ${\bf E}$, the second one by ${\bf H}$, and integrating in space over $\R^3$ the sum of the two resulting identities, one obtains, after integration by parts to make disappear the curl terms, 
\begin{equation} \label{idEnergy}
	\ds \frac{d}{dt} \, {\cal E}(t) + \int_{\R^3} \big(  \partial_t {\bf P} \cdot {\bf E} +  \partial_t {\bf M} \cdot {\bf H} \big)({\bf x}, t) \,  d{\bf x}  = 0.
\end{equation}
Thus, for any $T > 0$,
\begin{equation} \label{varEnergy}
	{\cal E}(T) - {\cal E}(0) + \int_{\R^3} \Big[ \int_0^T \big(  \partial_t {\bf P} \cdot {\bf E} +  \partial_t {\bf M} \cdot {\bf H} \big)({\bf x}, t) \, dt \Big] \; d{\bf x}  = 0.
\end{equation}
We thus infer that a medium is physically passive if for any $({\bf E}, {\bf H}, {\bf P}, {\bf M})$ satisfying (\ref{Maxwell2}, \ref{Dispersive2}) 
\begin{equation} \label{varEnergybis}
 \int_{\R^3} \int_0^T \big(  \partial_t {\bf P} \cdot {\bf E} +  \partial_t {\bf M} \cdot {\bf H} \big)({\bf x}, t) \, dt \; d{\bf x}  \geq 0.
\end{equation}
Considering \eqref{varEnergybis}, we introduce a stronger notion of passivity  (see e.g. \cite{Ber-11,Wel-jon-Yeh-16,cas-mil-17,cas-kach-jol-17}) which does not refer to Mawxell's equations anymore, but only the constitutive laws. 
\begin{Def}(Passivity) \label{Passivity}
	\noindent  A material is {passive} if, when $({\bf E}, {\bf P})$ and $({\bf H}, {\bf M})$ are related by the constitutive laws \eqref{Dispersive2}, for any $\bf x$ and $T>0$, 
	\begin{equation} \label{varEnergy2}
	\int_0^T \partial_t {\bf P}({\bf x}, t) \cdot {\bf E}({\bf x}, t)  \, dt  \geq 0 \quad \mbox{ and } \quad \int_0^T \partial_t {\bf M}({\bf x}, t) \cdot {\bf H} ({\bf x}, t) \, dt \geq 0.
	\end{equation}
\end{Def}
\noindent Of course, a passive medium in the sense of the above definition is a fortiori physically passive. The reverse statement seems to be often admitted in the physical community. For the authors of this chapter, it is however an open question (see the section 1.3 of \cite{cas-jol-ros-22} for more details). \\ [12pt]
One of the great interests of this notion of passivity of Definition \ref{Passivity} is that it can be traduced in terms of properties of the functions $\omega \mapsto \varepsilon ({\bf x}, \omega)$ and $\omega \mapsto \mu({\bf x}, \omega)$. This leads to introduce the important
notion of {Herglotz} function (se for e.g. \cite{Nev-22,Ges-00}).
\begin{Def} {(Herglotz function)} \label{def_herg} A Herglotz function is a function  $f(\omega) : \C^+ \rightarrow \C$, analytic in $\C^+$,  whose image is included in the closure of $\, \C^+$, i.e. 
	\begin{equation} \label{propHerglotz}
		\operatorname{Im} \, \omega > 0 \quad \Longrightarrow \quad  \operatorname{Im} \, f(\omega) \geq 0 .
	\end{equation}
\end{Def} 
\noindent 
\begin{Rem} \label{rem_herg} Note that, by the open mapping theorem (see e.g. \cite{Rud-87}), as soon an Herglotz function in not constant, the inequality \eqref{propHerglotz} is strict, i.e. $\operatorname{Im} \, f(\omega) > 0 $.  \end{Rem} 
\noindent Since $ \operatorname{Im} (1/f) = -  \operatorname{Im}(f) / |f|^2$, Herglotz functions are invariant by the transformation $f \rightarrow -1/f$.   Let us mention two properties which   will be particularly useful later. 
\begin{Lem} \label{lem_poles} [Poles and zeros]  Let $f$ be a non zero Herglotz function that can be extended to the whole complex half-space into a meromorphic function, still denoted $f$, the poles and zeros of $f$ are located in the half-space $\overline{\C^-} := \big \{ \operatorname{Im} \, z \leq 0 \}$ \end{Lem}
\begin{proof}
	The fact that $\operatorname{Im}\, f(\omega) > 0 $ in ${\C}^+$ imply that the zeros of $f$ belong to  $\overline{\C^-}$. For the poles, it suffices to reason with $-1/f$.
	\end{proof} 
\begin{Lem} \label{lem_positivity} [Positivity property] Let $f(\omega)$ be a non constant  Herglotz function that can be extended as a meromorphic function in some half-space 
	$\C_\alpha$ with $\alpha < 0$  (note that $\R \subset \C_\alpha$) which takes real values along the real axis:
\begin{equation} \label{hyp_real}
	\forall \; \omega \in \R \setminus {\cal P}_\R, \quad f(\omega) \in \R.
\end{equation} 
where ${\cal P}_\R$ denotes the discrete set of real poles of $f$. Then $f(\omega)$ satisfies the positivity property
\begin{equation} \label{positivity}
	\forall \; \omega \in \R \setminus {\cal P}_\R, \quad  f'(\omega)>0,
\end{equation} 
meaning that, between two consecutive poles in ${\cal P}_\R$, $\omega \mapsto f(\omega)$ is strictly increasing. 
Furthermore, the real poles and real zeros are of multiplicity one.
\end{Lem} 
\begin{proof} Let $ \eqref{positivity}$. 
As $f$ is analytic at $\omega_0$ and real valued on $\R \setminus {\cal P}_\R$, one shows  by induction via the Taylor expansion of  $f$ at $\omega_0$ expressed on the real  line  that $f^{(n)}(\omega_0)\in \bbR, \, \forall  n\in \mathbb{N}$. \\[6pt]
Let $\delta > 0$. When $\delta \rightarrow 0$, as $f(\omega_0), \,  f'(\omega_0) \in \R$, 
	$$
	f(\omega_0 + \rmi \, \delta) = f(\omega_0) + \rmi \, \delta \, f'(\omega_0) + o(\delta) \quad \mbox{  implies } \quad	\operatorname{Im}\, f(\omega + \rmi \, \delta)= \, \delta \, f'(\omega) +o(\delta).
	$$
As  $\operatorname{Im} \,  f(\omega + \rmi \, \delta)> 0$, one has necessary $f'(\omega_0)\geq0$. To conclude to \eqref{positivity}, we need to  prove that $f'(\omega_0)\neq 0$. To this aim, one assumes by contradiction  that $f'(\omega_0)=0$. As $f'$ is analytic on an open neighborhood $\mathcal{U}_{\omega_0}$  of  $\omega_0$, it is not zero on $\mathcal{U}_{\omega_0}$.  (Indeed, because of the analytic continuation principle,  $f$ is not constant on $\mathcal{U}_{\omega_0}$ since  $f$ is meromorphic and  not constant on $\bbC_{\alpha}$ by assumption). Thus, its  zeros in $\mathcal{U}_{\omega_0}$  are isolated. Hence,   the sequence $(f^{(m)}(\omega_0))_{m\geq 2}$ is different from $0$ and there exists  $n\geq 2$  such that 
$$
f(\omega)-f(\omega_0)= \frac{f^{(n)}(\omega_0)}{n!} (\omega-\omega_0)^n +o((\omega-\omega_0)^n) \  \mbox{ as }  \ \omega \to \omega_0 \ \mbox{ with }   \ \frac{f^{(n)}(\omega_0)}{n!}\in \bbR \setminus \{ 0\}.
$$
Thus,  for $\omega\in \bbC^+$ given by  $\omega- \omega_0= \rho \rme^{\rmi \theta}$ with $\rho>0$ sufficiently small and $\theta\in (0, \pi)$, one gets
$$
\operatorname{Im}f(\omega)=\frac{f^{(n)}(\omega_0)}{n!} \rho^n \sin(n \, \theta)+ o(\rho^n) \ \mbox{ as }  \ \omega \to \omega_0.
$$
As $n \geq 2$, $\sin (n \theta)$ describes $(-1,1)$ when $\theta$ varies in $(0, \pi)$ which leads to a contradiction since $\operatorname{Im}f(\omega)$ would take a negative value for some $\omega\in \bbC^+$ which is impossible since $f$ is a Herglotz function. Hence, one concludes to \eqref{positivity}. \\[6pt]
Thus, as by  \eqref{positivity},  $f'(\omega)>0$ for $\omega \in \R \setminus {\cal P}_\R$,  the real zeros of  $f$ are of multiplicity one. Applying the same reasoning to  to Herglotz function $-1/f$ (which satisfies also the  assumptions of this Lemma) gives that the real poles  (i.e.  the zeros of $-1/f$) of $f$ are also of multiplicity one.
	\end{proof}
\noindent The fundamental result of this section is given by the \\ [-18pt]
\begin{Thm}
A necessary and sufficient condition for a material obeying the constitutive laws \eqref{Dispersive} to be passive is that
\begin{equation} 
\omega \mapsto \omega \, \varepsilon ({\bf x}, \omega) \mbox{ and } \omega \mapsto \omega \, \mu({\bf x}, \omega) \mbox{ are Herglotz functions.}
\end{equation}
\end{Thm}
\noindent A general (but very technical) proof is done  in \cite{Zemanian} in the context  of  distributions valued in a Banach spaces.
A simpler proof of this theorem is done in detail in the case of Example \ref{Exconv} where the causal kernels $(\chi_e, \chi_m)$ belong to $W^{1,1}(\R)$ in which case $\omega \mapsto \omega \, \varepsilon ({\bf x}, \omega)$ and $\omega \mapsto \omega \, \mu({\bf x}, \omega)$ are analytic in $\C^+$ and continuous on $\overline{\C^+}$ (which implies directly ({\bf CP}) with $\alpha = 0$) and tend to $0$ at $\infty$ (by Riemann-Lebesgue theorem). 
The key steps of the proof (for $\varepsilon ({\bf x}, \omega)$) are then (we remain rather formal below, technical details are treated in \cite{cas-kach-jol-17}): \\ [5pt] 
(i) One first shows that $\operatorname{Im} \big(\omega \, \varepsilon ({\bf x}, \omega)\big)\geq 0$  for $\omega \in \R$. Indeed, consider for instance ${\bf E}$ compactly supported in time. Using both the causality ({\bf CP}) and reality ({\bf PR})  principles  as well as  Plancherel's theorem, we get (for $T$ large enough  so that $[0,T]$ contains the support of $\bE$):
	$$
\left|	\begin{array}{lll}
\ds 	\int_0^T \partial_t {\bf P}({\bf x}, t) \cdot {\bf E}({\bf x}, t)  \, dt & = & \ds  - \int_{\R} \rmi \, \omega \,  \widehat{\bf P}({\bf x}, \omega) \cdot \overline{\widehat{\bf E}({\bf x}, \omega)}   \, d \omega \\ [12pt]
	& = & \ds - \int_{\R} \rmi \, \omega \, \chi_e({\bf x}, \omega) \, \big|\widehat{\bf E}({\bf x}, \omega)\big|^2 \, d \omega 
	\\ [12pt]
	\quad \quad (\mbox{as } {\bf E}, {\bf P} \mbox{ are real}) & = & \ds \int_{\R} \operatorname{Im}  \big(\omega \, \chi_e({\bf x}, \omega)\big) \;\, \big|\widehat{\bf E}({\bf x}, \omega)\big|^2 \, d \omega .
	\end{array} \right.
	$$
	As the above in true for all compactly supported fields, by density
	$$
	\forall \; \omega \in \R, \quad \operatorname{Im}  \big(\omega \, \varepsilon ({\bf x}, \omega) \big) =	\operatorname{Im}  \big(\omega \, \chi_e ({\bf x}, \omega) \big) \geq 0
	$$
\noindent \hspace*{-0.2cm} (ii) One then shows that $\operatorname{Im} \big(\omega \, \varepsilon ({\bf x}, \omega)\big) > 0$  in $\C^+$.  Indeed, let us set 
$$u(x,y) := \operatorname{Im}  \big( (x+\rmi y) \, \chi_e (\cdot, x+\rmi y)\big) \mbox{ for } (x,y) \in  \R \times \R^+.$$ 
\noindent \hspace*{-0.2cm} The function $u$ is harmonic in the upper-half plane $\bbR\times \bbR_+$ (as the imaginary part of an analytic function in $\bbC^+$), positive  $\partial (  \R \times \R^+)= \R \times \{ 0\}$ (i.e. on the real line) and tends to $0$ at $\infty$. By the maximum principle $u(x,y) \geq 0$ in $ \R^2_+$, which implies $\operatorname{Im}  \big(\omega \, \varepsilon ({\bf x}, \omega)\big) = u(x,y) + \varepsilon_0\, \operatorname{Im} \, \omega > 0$.

\section{Maxwell's equations in general passive media} \label{Generalpassive}
In this section, for the simplicity of the presentation, we assume that the permittivity and permeability are independent of ${\bf x}$, i. e. that the medium is homogeneous
\begin{equation} \label{homedium} 
	\varepsilon({\bf x}, \omega) = \varepsilon(\omega), \quad 	\mu({\bf x}, \omega) = \mu(\omega).
	\end{equation}
	However, all what follows in this section extends to the heterogeneous case without any difficulty.
\subsection{A representation of $\varepsilon(\omega)$ and $\mu(\omega)$ in passive media}  \label{Representationforula}
In this section, we assume the familiarity of the reader with basics of measure theory on $\mathbb{R}$ \cite{Rud-87}. What follows is based on the version of the well-known Nevanlinna's representation theorem for Herglotz functions  satisfying the symmetry property 
\begin{equation}  \label{symmetry} 
\forall \; \omega \in \C^+, \quad f(-\overline{\omega}) = - \overline{f(\omega)}
	\end{equation} 
	\begin{Thm}\label{thm.herglotz}
		Let $f(\omega)$ be a Herglotz function satisfying \eqref{symmetry}, then there exists  a unique positive regular measure $\nu$ on the set $\mathcal{B}(\bbR)$ of Borelians of $\R$ satisfying
	\begin{equation} \label{propmes} 
			\forall \, E \in \mathcal{B}(\R), \quad \nu(-E) = \nu(E) \quad \mbox{(evenness),} 
	\quad  \int_{\R} \
				\frac{\md  \nu(\xi)}{1+\xi^2} < + \infty, \quad \mbox{(finiteness}),
			\end{equation} 
such that $f(\omega)$ admits the following integral representation
	\begin{equation}\label{expepsmu}
\forall \; \omega \in \C^{+}, \quad 	f(\omega)= \omega \Big( \alpha +  \displaystyle \int_{\R}\frac{\md \nu( \xi)}{\xi^2 - \omega^2} \; \Big), \quad \mbox{where }  \alpha:= \lim_{y \rightarrow +\infty}f(\rmi \, y)/{\rmi \, y}
	\end{equation}
	\end{Thm}
\noindent	A complete proof of this theorem can be found in \cite{cas-kach-jol-17} (Theorem 4.5 and appendix B). It is worthwhile mentioning that 	$\nu$ is related to $f$  (for any $a < b$) as follows
		\begin{equation}  \label{measure}
				 \nu(a)=\lim_{\eta\to 0^{+}}\eta \;  \operatorname{Im} \,  f(a+  \rmi \eta), \quad \quad 
			\nu\big([a,b]\big)+ \nu\big( ]a,b[ \big) =\lim_{\eta\rightarrow 0^{+}} \frac{2}{\pi} \int_{a}^{b}  \operatorname{Im}  \, f(x+ \rmi \, \eta) \, dx .
	\end{equation}
	The above formula emphasize the key role of the limit of the imaginary part $\operatorname{Im}  \, f$ when one approaches the real axis from above.
	\begin{Rem} \label{remNevanlinna2} 
		The formulas (\ref{measure}) provide the measure of any interval $[a,b)$, $(a,b]$, $(a,b)$ or $[a,b]$. Thus it defines completely $\nu$ as a borelian measure \cite{Rud-87}.
	\end{Rem}	
	\noindent As, in a passive media, $\omega \, \varepsilon(\omega)$ and $\omega \, \mu(\omega)$ are Herglotz functions satisfying \eqref{symmetry} (because of  $({\bf RP})$), one immediately deduces from Theorem \ref{thm.herglotz}  the following
	\begin{Cor}\label{cor.herglotz}
		Let $\varepsilon(\omega)$ and $\mu(\omega)$ be the permittivity and  permeability of a homogeneous 
		passive medium. Then, there exists two unique regular  positive  measures $\nu_e$ and $\nu_m$ on $\mathcal{B}(\bbR)$ satisfying \eqref{propmes} such that 
		\begin{equation}\label{expepsmubis}
			\varepsilon(\omega)=\varepsilon_0 \; \Big( 1+ \displaystyle \int_{\R}\frac{\md \nu_e( \xi)}{{\xi^2 - \omega^2}} \Big) , \quad \mu(\omega)=\mu_0 \; \Big( 1+ \displaystyle \int_{\R}\frac{\md \nu_m( \xi)}{{\xi^2 - \omega^2}} \Big)\ \mbox{ for } \omega \in \C^{+}.
		\end{equation}
	\end{Cor}
\subsection{Augmented PDE model for Maxwell's equations in passive media} \label{Augmented}
At least from the theoretical point of view, one of the interests of the representation formulas \eqref{expepsmu} is that they permit to rewrite the Mawell's equations in a passive media as a ``standard" evolution problem, modulo the introduction of an additional variable $\xi \in \R$,  namely the integration variable in \eqref{expepsmubis}, and additional auxiliary unknowns, that are functions of $({\bf x}, t, \xi)$ and permit to represent the constitutive law in terms of ODEs in time. Such an approach was developed in \cite{Tip}, \cite{GralakTip}, \cite{Figotin},  \cite{cas-kach-jol-17} following the pioneering
work of Lamb \cite{Lamb}.\\ [12pt]
From and the definition \eqref{defsus} of the susceptibilities $\chi_e$ and $\chi_m$, we have
$$
 \chi_e (\omega)  = \int_{\R}\frac{\md \nu_e( \xi)}{{\xi^2 - \omega^2}} , \quad  \chi_m (\omega)  =  \int_{\R}\frac{\md \nu_m( \xi)}{{\xi^2 - \omega^2}} .
$$
Therefore, the constitutive laws $	\widehat{\bf P}({\bf x}, \omega) =  \varepsilon_0 \,{\chi}_e(\omega)\, \widehat{\bf E}({\bf x},\omega)$ and $\widehat{\bf M}({\bf x}, \omega) =  \mu_0 \,{\chi}_m(\omega)\, \widehat{\bf H}({\bf x},\omega)$ in the frequency domain, see \eqref{Dispersive2} can be rewritten as 
\begin{equation} \label{Dispersive2bis}
\left|	\begin{array}{lll}
\ds \widehat{\bf P}({\bf x}, \omega) = \varepsilon_0 \int_{\R} \widehat\bbP({\bf x}, \omega ;\xi) \,  \md \nu_e( \xi), & \quad  \big(\xi^2 - \omega^2\big) \, \widehat\bbP({\bf x}, \omega ;\xi) = \varepsilon_0 \,  \widehat{\bf E}({\bf x}, \omega),
\\ [12pt]
\ds \widehat{\bf M}({\bf x}, \omega) =\mu_0  \int_{\R} \widehat\bbM({\bf x}, \omega ;\xi) \,  \md \nu_m( \xi), & \quad  \big(\xi^2 - \omega^2\big) \, \widehat\bbM({\bf x}, \omega ;\xi) = \mu_0 \, \widehat{\bf H}({\bf x}, \omega).
\end{array} \right.
\end{equation}
Using the property  \eqref{propFLT} of the Fourier-Laplace transform, these can be rewritten in time as 
\begin{equation} \label{Dispersive2time}
	\left|	\begin{array}{lll}
		\ds {\bf P}({\bf x}, t) = \varepsilon_0 \int_{\R} \bbP({\bf x}, t ;\xi) \,  \md \nu_e( \xi), & \quad   \partial_t^2 \bbP({\bf x}, t ;\xi) + \xi^2 \,  \bbP({\bf x}, t ;\xi) \; \; = {\bf E}({\bf x}, t),
	\\ [12pt]
		\ds {\bf M}({\bf x}, t) = \mu_0 \int_{\R} \bbM({\bf x}, t ;\xi) \,  \md \nu_m( \xi), & \quad   \partial_t^2 \bbM({\bf x}, t ;\xi) + \xi^2 \, \bbM({\bf x}, t ;\xi) =  {\bf H}({\bf x}, t),
	\end{array} \right.
\end{equation}
where the ODEs defining $ \bbP$ and $ \bbM$ must be completed with vanishing initial data
\begin{equation} \label{CIPM} 
\bbP({\bf x}, 0 ;\xi) =  \partial_t \bbP({\bf x}, 0 ;\xi) = 0, \quad \bbM({\bf x}, 0 ;\xi) = \partial_t \bbM({\bf x}, 0 ;\xi) =0.
\end{equation}
As a consequence, the 3D Cauchy problem for Maxwell equation in a passive medium can be formulated as follows 
 \\ [12pt]
{Find } $\ds  \left\{ \begin{array}{lll} {\bf E}({\bf x}, t) : \R^3 \times \R^+ \rightarrow \R^3, &  \quad 	\bbP({\bf x}, t ;\xi) : \R^3  \times \R^+ \times \R \rightarrow \R^3, \\[12pt]
{\bf H}({\bf x}, t) : \R^3 \times \R^+ \rightarrow \R^3,	  &  \quad \bbM({\bf x}, t ;\xi) : \R^3  \times \R^+ \times \R \rightarrow \R^3, \end{array} \right. $
\\[12pt]
that satisfy the integro-differential system 
	\begin{equation} \label{Lorentzsystemgene}
		\hspace*{-0.3cm}\left\{	\begin{array}{lll}
			\ds 
			\ds	\varepsilon_0 \, \partial_t {\bf E} + {\bf rot} \, {\bf H} +   \varepsilon_0 \int_{\R} \partial_t \bbP (\cdot ,\cdot;\xi) \, \md \nu_e(\xi) = 0, & \quad ({\bf x}, t) \in \R^3 \times \R_+,\\[12pt]
			\ds	\mu_0 \, \partial_t {\bf H} - {\bf rot} \, {\bf E} + \mu_0 \int_{\R} \partial_t {\bbM}(\cdot ,\cdot;\xi)  \, \md \nu_m(\xi) = 0, & \quad ({\bf x}, t) \in \R^3 \times \R_+, \\[12pt]
			\partial_t^2 \bbP  (\cdot ,\cdot;\xi)  + \xi^2 \, \bbP (\cdot ,\cdot;\xi)  =    {\bf E}, & \quad ({\bf x}, t ;\xi) \in \R^3  \times \R_+ \times \R, \\[12pt] \partial_t^2 \bbM  (\cdot ,\cdot;\xi)  + \xi^2 \, \bbM  (\cdot ,\cdot;\xi)  =  {\bf H}, & \quad ({\bf x}, t ;\xi) \in \R^3  \times \R_+ \times \R ,
		\end{array} \right.
	\end{equation}
	completed by the initial conditions \eqref{CIPM}  as well as 
\begin{equation} \label{CIEH} 
	{\bf E}({\bf x}, 0) = 	{\bf E}_0({\bf x}) , \quad 	{\bf H}({\bf x}, 0) = 	{\bf H}_0({\bf x}).
\end{equation}
\begin{Rem}
	Owing to the integral representation of ${\bf P}$ and ${\bf M}$ in \eqref{Dispersive2time}, the fields $\bbP$ and $\bbM$ can be reinterpreted as polarization and magnetization densities respectively.
	\end{Rem}
\noindent One can easily obtain an energy conservation result for the augmented system \eqref{Lorentzsystemgene}: 
\begin{Thm} \label{thm.Energygene} Any smooth enough solution of  (\ref{Lorentzsystemgene}) satisfies the energy identity  
	\begin{equation} \label{energyidentitygene}
			\ds \frac{d}{dt} \, {\cal E}_{cons}(t) = 0,  \quad \mbox{where } \quad {\cal E}_{cons}(t) :=  {\cal E}(t) + {\cal E}_{e}(t) + {\cal E}_{m}(t),
	\end{equation}
	with the electromagnetic energy $ {\cal E}$ defined by and the additional energies ${\cal E}_{e}$ and ${\cal E}_{m}$ by 
		\begin{equation} \label{additionalenergies}
		\left\{	\begin{array}{l}
\ds	{\cal E}_{e}(t) :=  \frac{ \varepsilon_0}{2}\int_{\R} \int_{\R^3} \Big( \, |\partial_t \bbP({\bf x}, t ;\xi)|^2  +  \xi^2 \, |\bbP({\bf x}, t ;\xi)|^2 \, \Big) \; d{\bf x} \; \md \nu_e(\xi),\\[15pt]
			\ds  
			{\cal E}_{m}(t) :=   \frac{ \mu_0}{2}  \int_{\R} \int_{\R^3} \Big( \, |\partial_t \bbM({\bf x}, t ;\xi)|^2  +  \xi^2 \, |\bbM({\bf x}, t ;\xi)|^2 \, \Big) \; d{\bf x} \; \md \nu_m(\xi), 
		\end{array} \right.
	\end{equation}
\end{Thm}

\begin{proof} From the ODE's in \eqref {Dispersive2time}, multiplied respectively by $\partial_t \bbP$ and $\partial_t \bbM$, we have
$$
		\left|	\begin{array}{lll}
\partial_t \Big( \frac{1}{2} \, |\partial_t \bbP({\bf x}, t ;\xi)|^2 +  \frac{1}{2} \,\xi^2 \,  |\bbP({\bf x}, t ;\xi)|^2 \; \Big) = \partial_ t \bbP({\bf x}, t ;\xi) \cdot  {\bf E}({\bf x}, t)
		\\ [12pt]
\partial_t \Big( \frac{1}{2}  \,|\partial_t \bbM({\bf x}, t ;\xi)|^2 +  \frac{1}{2} \, \xi^2 \,  |\bbM({\bf x}, t ;\xi)|^2 \;  \Big) =  \partial_ t \bbM({\bf x}, t ;\xi) \cdot  {\bf H}({\bf x}, t)
\end{array} \right.
$$
Thus, after integration over $\xi$, using the formulas for ${\bf P}$ and ${\bf M}$ in \eqref {Dispersive2time}, we obtain 
$$
\left|	\begin{array}{lll}
\ds  \partial_t {\bf P}({\bf x}, t)  \cdot {\bf E}({\bf x}, t) =  \frac{ \varepsilon_0 }{2} \int_{\R}  \partial_t \Big( \, |\partial_t \bbP({\bf x}, t ;\xi)|^2  +  \xi^2 \, |\bbP({\bf x}, t ;\xi)|^2 \, \Big)  \; \md \nu_e(\xi),
	\\ [12pt]
\ds  \partial_t {\bf M}({\bf x}, t)   \cdot {\bf H}({\bf x}, t)= \frac{\mu_0}{2} \int_{\R}  \partial_t  \Big( \, |\partial_t \bbM({\bf x}, t ;\xi)|^2  +  \xi^2 \, |\bbM({\bf x}, t ;\xi)|^2 \, \Big)  \; \md \nu_e(\xi).
\end{array} \right.
$$
Finally, if suffices to substitute the above into the identity \eqref{idEnergy} to obtain \eqref{energyidentitygene}.
\end{proof}
\noindent The identity \eqref{energyidentitygene} permits to recover the physical passivity of the medium, i.e.  in the energetic sense \eqref{propEnergy}. Indeed, for any $T > 0$, 
$$
 {\cal E}(T) \leq  {\cal E}_{cons}(T) =  {\cal E}_{cons}(0) =  {\cal E}(0),
$$
where the first inequality results from the positivity of ${\cal E}_{e}$ and ${\cal E}_{m}$, the second equality from \eqref{energyidentitygene} and the last one from \eqref{CIPM} .
\subsection{Abstract reformulation. Existence and uniqueness results.} \label{sec-abstract} 
For technical reasons, see remark \ref{remcondsup}, we shall need  in this section the following assumption (that is stronger than the finiteness assumption in \eqref{propmes}) 
\begin{equation} \label{finiteness2}
\boldsymbol{\nu}_e := \int_{\R} \md \nu_{e}(\xi) < + \infty, \quad \boldsymbol{\nu}_m := \int_{\R} \md \nu_{m}(\xi) < + \infty.
\end{equation} 
Modulo the introduction of the additional unknowns $\dot\bbP := \partial_t \bbP$ and $\dot\bbM := \partial_t \bbM$, (\ref{Lorentzsystem}) can be rewritten as a first order evolution problem for ${\bf U} = \big( {\bf E}, {\bf H}, \bbP, \dot\bbP, \bbM, \dot\bbM\big)$ 
\begin{equation}\label{eq.schro}
	\frac{\md \, {\bf U}}{\md\, t} + i\,\mathbb{A} \, {\bf U}=0.
\end{equation}
More precisely, the unknown ${\bf U}$ is searched as a function of $t$ with values in the 
 Hilbert space:
\begin{equation} \label{spaces}
\left\{ \begin{array}{lll}
	\boldsymbol {\mathcal{H}}= \bL^2(\mathbb{R}^3) \times \bL^2(\mathbb{R}^3)   \times  \boldsymbol {\mathcal{  V}_{e}} \times \boldsymbol {\mathcal{H}_{e}}\times  \boldsymbol {\mathcal{V}_{m}} \times \boldsymbol { \mathcal{H}_{m} }, & \bL^2(\mathbb{R}^3)= L^2(\R^3)^3, & \quad (i)
\\ [12pt]
\boldsymbol{\mathcal{H}}_{s} = L^2\big(\R, \bL^2(\mathbb{R}^3); \md \nu_{s} \big), \
\boldsymbol {\mathcal{V}}_{s}= L^2\big(\R, \bL^2(\mathbb{R}^3) ;\,  \xi^2 \,  \md \nu_{s} \big), &   s=e ,m.& \quad (ii)
\end{array} \right.
\end{equation}
We assume  that the spaces $( \bL^2(\mathbb{R}^3), \boldsymbol{\mathcal{H}}_{s}, \boldsymbol{\mathcal{H}}_{s})$ are equipped with their natural norm and we equip $\boldsymbol {\mathcal{H}}$ with a norm which is (up to a $1/2$ factor) given by the  energy ${\cal E}_{tot}$, see \eqref{energyidentitygene}:
\begin{equation}\label{eq.defprodscal}
	\left| \; 	\begin{array}{lllll}
		\ds \|{\bf U}\|^2_{\boldsymbol{\cal H}} & =  &  &\ds \varepsilon_0 \, \int _{\R^3} |{\bf E} (\bx)|^2  \, \md{\bf x}    & \ds +  \varepsilon_0  \int_{\R} \int_{\R^3} \big(\xi^2 \, | \bbP(\bx,\xi)|^2 +  |\dot \bbP(\bx,\xi)|^2\big) \, \md {\bf x}\,  \md \nu_e(\xi) \\ [14pt] 
		& &+ &
		\ds  \mu_0 \,\int_{\R^3}  |{\bf H}(\bx)|^2  \, \md{\bf x} &\ds + \mu_0  \int_{\R} \int_{\R^3} \big(\xi^2 \, | \bbM(\bx,\xi)|^2 +  |\dot \bbM(\bx,\xi)|^2\big) \, \md {\bf x}\,  \md \nu_m(\xi).
	\end{array} \right.
\end{equation}
The operator $\bbA$ is an unbounded operator on $\boldsymbol{\mathcal{H}}$ with domain
$$
\left\{ \begin{array}{lll}
D(\bbA) =  H({\bf rot}, \R^3) \times   H({\bf rot}, \R^3)   \times  \boldsymbol{\mathcal{D}_{e}} \times  \boldsymbol {\mathcal{\dot D}_{e}} \times  \boldsymbol{\mathcal{D}_{m}} \times \boldsymbol {\mathcal{\dot D}_{m}}, & \quad ,
	\\ [12pt]
\boldsymbol{\mathcal{D}}_{s} := \big \{ \bbP \in  \boldsymbol {\mathcal{V}_{s} } \,  /  \, \xi^2 \, \bbP \in  \boldsymbol{\mathcal{H}_{s}} \big\},  \;  \boldsymbol {\mathcal{\dot D}_{s}} := \boldsymbol{\mathcal{H}_{s}} \cap  \boldsymbol {\mathcal{V}}_{s},
 &  \quad s=e ,m.
\end{array} \right.
$$ 
where $ H({\bf rot}, \R^3) :=\{\bu \in \bL^2(\mathbb{R}^3)\mid \nabla \times \bu \in  \bL^2(\mathbb{R}^3)  \}$.
 $\bbA$ is defined in block form by
\begin{equation}\label{eq.opA}
\bbA := \ \rmi\, \begin{pmatrix}
	0 &\varepsilon_0^{-1}\,\mbox{\bf rot} & 0&-\mathbb{M}_e & 0 & 0\\[2pt]
	- \mu_0^{-1}\,\mbox{\bf rot} & 0 &0 & 0& 0& -\mathbb{M}_m \\[2pt]
	0 & 0 &0  & 1&0 &0\\[2pt]
	\mathbb{M}_e^* & 0 &-\, \xi^2  &0 &0 &0\\[2pt]
	0 & 0 &0  & 0 &0 & 1\\[2pt]
	0 &\mathbb{M}_m^*& 0 & 0 & -\, \xi^2   & 0
\end{pmatrix},
\end{equation}
where $\boldsymbol{\cal M}_s \in B(\boldsymbol{\cal H}_s, \bL^2(\mathbb{R}^3) )$ and its adjoint  $\boldsymbol{\cal M}_s^*  \in B(\bL^2(\mathbb{R}^3) ,\boldsymbol{\cal H}_s), s = e,m$ are defined by 
$$
\forall \; (\mathbb{X}, {\bf F}) \in \boldsymbol{\cal H}_s \times \bL^2(\mathbb{R}^3) , \quad  \big({\mathbb M}_s \mathbb{X}\big)({\bf x})  := \int_{\bbR} \bbX ({\bf x}, \xi) \;  \rmd  \nu_s(\xi), \quad \big({\mathbb M}_s^* {\bf F}\big)({\bf x}, \xi) = {\bf F}({\bf x}).
$$

\begin{Rem} \label{remcondsup} Notice that  (\ref{finiteness2}) is needed for defining  $\boldsymbol{{\cal M}_s}$.  Indeed, the condition (\ref{finiteness2}) means that $1 \in L^2(\R, \md \nu_s(\xi))$ with
$L^2(\R, \md \nu_s(\xi))$-norm equal to $\boldsymbol{\nu}_s$. Thus, by Cauchy-Schwarz,
		$$
\big| \big({\mathbb M}_s \bbP\big)({\bf x})\big|^2  \leq \boldsymbol{\nu}_s^2  \, \Big(\int_{\R} \bbP ({\bf x}, \xi) \md \nu_{s}(\xi)\Big)^2 \quad\Longrightarrow \quad  \|{\mathbb M}_s  \bbP\|_{\bL^2(\bbR^3)} \leq  \boldsymbol{\nu}_s \, \|\bbP\|_{\boldsymbol{\cal H}_s}.
$$
Similarly, for the adjoint operator ${\mathbb M}_s^*$, one observes that $\big\|{\mathbb M}_s^* {\bf E}\big\|_{\boldsymbol{\cal H}_s} =\boldsymbol{\nu}_s \, \big\| {\bf E} \|_{\bL^2(\mathbb{R}^3) } $. 
\end{Rem}
\noindent We leave to the reader the exercise to show that
the operator $\bbA$ is self-adjoint in $\boldsymbol{\mathcal{H}}$ for the inner product associated to the norm \eqref{eq.defprodscal}. \\ [12pt]
		 From semi-group theory, or  Hille-Yosida's theorem \cite{Paz-83}, $\bbA$ is thus the generator
		of a $\mathcal{C}^0$ unitary group and  the evolution problem associated to (\ref{eq.schro}) completed with the initial condition ${\bf U}(0) = {\bf U}_0 \in D(\bbA)$ admits a unique solution
\begin{equation} \label{regU}
		{\bf U}(t) :=  e^{-\rmi\bbA t} \, {\bf U}_0  \in C^1(\R^+; {\cal H}) \cap C^0(\R^+; D(\bbA)). 
\end{equation}
		In particular, going back to the initial conditions and translating \eqref{regU} in terms of the original unknown $({\bf E}, {\bf H}, \bbP, \bbM)$, we deduce that for $({\bf E}_0, {\bf H}_0) \in H({\bf rot}, \R^3)^2$ the problem (\ref{Lorentzsystemgene}, \ref{CIPM}, \ref{CIEH}) admits a unique solution satisfying
\begin{equation} \label{regEHPM}
\left\{\begin{array}{lll}
	{\bf E},  {\bf H} \in C^1(\R^+;  \bL^2(\mathbb{R}^3)\big)  \cap C^0\big(\R^+;  H({\bf rot}, \R^3)\big) \\ [12pt]
	\bbP \in C^2(\R^+; \boldsymbol{\mathcal{H}}_e)  \cap C^1(\R^+; \boldsymbol{\mathcal{V}}_e)  \cap C^0(\R^+; \boldsymbol{\mathcal{D}}_e),  \\ [12pt]
	\bbM \in C^2(\R^+; \boldsymbol{\mathcal{H}}_m)  \cap C^1(\R^+; \boldsymbol{\mathcal{V}}_m)  \cap C^0(\R^+; \boldsymbol{\mathcal{D}}_m), 
	\end{array}  \right.
\end{equation}		

\begin{Rem} Since the operator $e^{-i\bbA t}$ is unitary in $\boldsymbol{\cal H}$, $\|{\bf U}(t)\|_{\boldsymbol{\cal H}}$ is constant in time which is nothing but the conservation of energy (cf. theorem \ref{thm.Energygene}).
\end{Rem}
\begin{Rem} The spectrum of the operator $\bbA$ is described in detail in \cite{cas-kach-jol-17}, section 4.3.5. \end{Rem}
\begin{Rem} \label{rem_techass} The technical assumption \eqref{finiteness2} is needed to put our problem in the framework of the semi-group theory by defining $\bbA$ as  a self-adjoint operator. However, without \eqref{finiteness2}, the existence result could be proved with another approach, Galerkin's method for instance. 
\end{Rem}

\section{Generalized Lorentz media}\label{sec-Gen-Lorentz}
\subsection{Definition of the models} \label{sec-Lorentz-def}
\noindent Lorentz materials  are obtained  with measures made of finite (even) sums  of Dirac measures:
\begin{equation}\label{LorentzMeasures}
	\nu_e = \frac{1}{2} \; \sum_{j = 1}^{N_e} {\Omega_{e,j}^2} \, \big( \delta_{\omega_{e,j}} + \delta_{-\omega_{e,j}}\big), \quad \nu_m = \frac{1}{2} \; \sum_{\ell = 1}^{N_m} {\Omega_{m,\ell}^2} \, \big( \delta_{\omega_{m,\ell}} + \delta_{-\omega_{m,\ell}}\big),
\end{equation}
where the $\{ \pm \, \omega_{e,j} \}$ and $\{ \pm \, \omega_{m,\ell} \}$ are called respectively the electric and magnetic resonances:
\begin{equation}\label{resonances}
0 \leq \omega_{e,1} <  \omega_{e,2} < \cdots < \omega_{e,N_e}, \quad 0 \leq \omega_{m,1} <  \omega_{m,2} < \cdots < \omega_{m,N_m}.
\end{equation}
In this case, the functions $\varepsilon(\omega)$ and $\mu(\omega)$ are rational functions 
\begin{equation} \label{Lorentzlaws}
	\varepsilon(\omega)	= \varepsilon_0 \; \Big( 1 + \sum_{j = 1}^{N_e} \frac{\Omega_{e,j}^2}{\omega_{e,j}^2 -\omega^2}\Big), \quad \mu(\omega)	= \mu_0 \; \Big( 1 + \sum_{\ell = 1}^{N_m} \frac{\Omega_{m,\ell}^2}{\omega_{m,\ell}^2 -\omega^2}\Big).
\end{equation}
\begin{Rem} \label{LorentzDrude} Standard Lorentz models correspond to $N_e = N_m = 1$ and are called Drude models if $\omega_{e,1} = \omega_{m,1} = 0$. By extension, each simple term in the sums \eqref{Lorentzlaws} is called a Lorentz term or Drude term, depending on the fact that the associated resonance is $0$ or not. 	\end{Rem}
\begin{figure}[h!] 
	\centering
		\includegraphics[height=5.5cm]{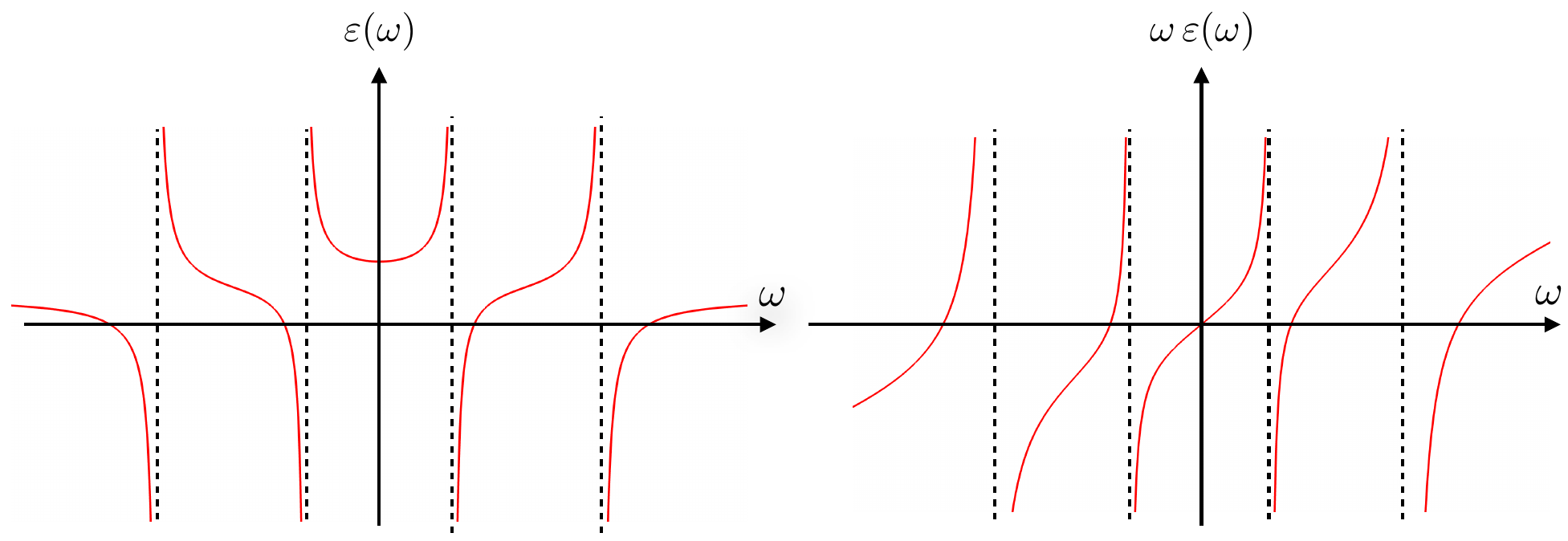} 
	\caption{graph of $\omega \mapsto \varepsilon(\omega)$ (left) and  positivity property of $\omega \mapsto \omega \, \varepsilon(\omega)$ (right) for $N_e = 2$ and $\omega_{e,1},  \,\omega_{e,2}> 0$.}
	\label{Fig1}
	\end{figure}
\noindent Clearly, the poles and zeros of $\varepsilon(\omega)$ are real. Apart from 0, that is a double pole if $\varepsilon(\omega)$ contains a Drude term, all other poles are  simple and interlace with the (simple) zeros on the real axis along which $\varepsilon(\omega)$ is real. We represent on figure \ref{Fig1}  the typical graph of $\omega \rightarrow \varepsilon (\omega)$ (when $\omega_{e,1}, \,\omega_{e,2} > 0$) and illustrate the positivity property of $\omega \, \varepsilon(\omega)$, see Lemma \ref{lem_positivity} applied to $\omega \, \varepsilon(\omega)$. \\ [10pt]
\noindent	Lorentz media are local media in the sense of Example \ref{Exloc}. Therefore, the corresponding evolution problem and can be rewritten  coupling of Maxwell's equations in the vacuum with a system of linear second order ODE's  (harmonic oscillators). More precisely, the general model of section \ref{Augmented}, see \eqref{Lorentzsystemgene}, takes the particular form: \\ [12pt]
{Find } $\ds  \left\{ \begin{array}{lll} {\bf E}({\bf x}, t) : \R^3 \times \R^+ \rightarrow \R^3, &  \quad 	\bbP_{\!j}({\bf x}, t) : \R^3  \times \R^+  \rightarrow \R^3,  1 \leq j \leq N_e\\[12pt]
	{\bf H}({\bf x}, t) : \R^3 \times \R^+ \rightarrow \R^3,	  &  \quad \bbM_\ell({\bf x}, t ;\xi) : \R^3  \times \R^+\! \times \R \rightarrow \R^3, 1 \leq \ell  \leq N_m,  \end{array} \right. $ 
\\[12pt]
satisfying the PDE system
	\begin{equation} \label{Lorentzsystem}
		\left\{	\begin{array}{lll}
			\ds	\varepsilon_0 \, \partial_t {\bf E} + {\bf rot} \, {\bf H} + \varepsilon_0 \, \sum_{j=1}^{N_e}\Omega_{e, j}^2  \, \partial_t  \bbP_j= 0, & \quad  \partial_t^2 \bbP_j + \omega_{e, j}^2 \,  \bbP_j = \,  {\bf E},\\[18pt]
		\ds \mu_0 \, \partial_t {\bf H} - {\bf rot} \, {\bf E}  + \mu_0 \,  \sum_{\ell = 1}^{N_m} \Omega_{m, \ell}^2 \, \partial_t \bbM_\ell = 0, 	 &  \quad \partial_t^2 \bbM_\ell + \omega_{m, \ell}^2 \, \bbM_\ell =  \,  {\bf H}.
		\end{array} \right.
	\end{equation}
	~\\[-4pt]
	completed with the initial conditions \eqref{CIEH} for $(\bE, \bH)$, and for the additional unknowns, 
\begin{equation} \label{CIPM2}  \bbP_{\!j}({\bf x},0) = \partial_t \bbP_{\!j}({\bf x},0)=0 \mbox{ and } \bbM_\ell ({\bf x},0) = \partial_t\bbM_\ell ({\bf x},0)=0 .
	\end{equation} 
Due to \eqref{LorentzMeasures}, the formula \eqref{Dispersive2time}  for the total polarization  and magnetization  becomes
\begin{equation}\label{Dispersive2time-disc}
{\bf P}=\eps_0\, \sum_{j=1}^{N_e} \Omega_{e,j}^2 \bbP_j, \quad {\bf M}=\mu_0\, \sum_{\ell=1}^{N_m} \Omega_{m,\ell}^2 \bbM_\ell.
\end{equation}
The conserved total energy ${\cal E}_{tot}(t)$  is still given by \eqref{energyidentitygene}, with this time, 
	\begin{equation} \label{additionalenergiesbis}
	\left\{	\begin{array}{l}
		\ds	{\cal E}_{e}(t) :=  \frac{1}{2} \, \varepsilon_0 \ \sum_{j=1}^{N_e}\Omega_{e, j}^2 \int_{\R^3} \Big( \, |\partial_t \bbP_{\!j}({\bf x}, t)|^2  +  \omega_{e,j}^2 \, |\bbP_{\!j}({\bf x}, t)|^2  \, \Big) \; d{\bf x},\\[15pt]
		\ds  
	\ds	{\cal E}_{m}(t) :=  \frac{1}{2} \, \varepsilon_0 \ \sum_{\ell=1}^{N_m}\Omega_{m, \ell}^2 \int_{\R^3} \Big( \, |\partial_t \bbM_\ell({\bf x}, t)|^2  +  \omega_{e,\ell}^2 \, |\bbM_\ell({\bf x}, t)|^2  \, \Big) \; d{\bf x}.
	\end{array} \right.
	\end{equation}
\begin{Rem} 
	More general materials  are obtained with infinite sums of Dirac discrete measures: 
	\begin{equation}\label{DiscreteMeasures}
		\nu_e = \frac{1}{2} \; \sum_{\ell = 1}^{+\infty} {\Omega_{e,j}^2} \,  \big( \delta_{\omega_{e,j}} + \delta_{-\omega_{e,j}}\big), \quad \nu_m = \frac{1}{2} \; \sum_{\ell = 1}^{+\infty} \, {\Omega_{m,\ell}^2} \big( \delta_{\omega_{m,\ell}} + \delta_{-\omega_{m,\ell}}\big).
	\end{equation}
	where $\omega_{e,\ell}$ and $\omega_{m,\ell}$ are two sequences of positive real numbers satisfying 
	\begin{equation}\label{Additional}
		\sum_{j = 1}^{+\infty} \frac {\Omega_{e,j}^2 }{1 + \omega_{e,j}^2}  < + \infty, \quad \quad 
		\sum_{\ell = 1}^{+\infty} \frac {\Omega_{m,\ell}^2}{1 + \omega_{m,\ell}^2}  < + \infty, 
	\end{equation}
	which corresponds the  finiteness condition \ref{propmes} in Corollary \ref{cor.herglotz}. $\varepsilon(\omega)$ and $\mu(\omega)$ are meromorphic functions defined by the following series,  whose convergence (outside poles) is ensured by (\ref{Additional}):
	\begin{equation} \label{InfiniteLorentzlaws}
		\varepsilon(\omega)	= \varepsilon_0 \; \Big( 1 + \sum_{j = 1}^{+ \infty} \frac{\Omega_{e,j}^2}{\omega_{e,j}^2 -\omega^2}\Big), \quad \mu(\omega)	= \mu_0 \; \Big( 1 + \sum_{\ell = 1}^{+ \infty} \frac{\Omega_{m,\ell}^2}{\omega_{m,\ell}^2 -\omega^2}\Big).
	\end{equation}
	An example of such a function is  $\varepsilon(\omega)=\varepsilon_0(1+a \, \omega^{-1}\tan \omega)$. Such functions appear naturally in the mathematical theory of metamaterials via high contrast homogenization \cite{Zhikov}, \cite{bouchitte2010homogenization}, \cite{Che-2018}. 
\end{Rem} 
\subsection{The abstract model} \label{sec-abstract2} 
As in section \ref{sec-abstract},  (\ref{Lorentzsystem}) can be rewritten as a first order evolution problem 
\begin{equation}\label{eq.schrobis}
	\frac{\md \, {\bf U}}{\md\, t} + i\,\mathbb{A} \, {\bf U}=0, \quad \mbox{with } {\bf U} = \big( {\bf E}, {\bf H}, \bbP, \dot\bbP, \bbM, \dot\bbM\big),
\end{equation}
 in the Hilbert space $\boldsymbol{\cal H}$ still defined by \eqref{spaces}(i) except that this time,  $\boldsymbol{\mathcal{H}}_{s} = 	\boldsymbol {\mathcal{V}}_{s} \equiv \big(\bL^2(\R^3)\big)^{N_s}$ for $s = e,m$, where the norm in $\boldsymbol{\cal H}$ is
\begin{equation}\label{eq.defprodscal2} 
	\left| \; 	\begin{array}{lll}
	\ds \|{\bf U}\|^2_{\boldsymbol{\cal H}} & =  & \ds  \varepsilon_0 \, \int_{\R^3}  |{\bf E}(\bx)|^2  \, \md{\bf x}    +  \varepsilon_0  \int_{\R^3} \sum_j \big(\omega_{e,j}^2 \,  | \bbP_j(\bx)|^2 +  |\dot \bbP_j(\bx)|^2\big) \, \md {\bf x}\  \\ [18pt] 
		& + & 
		\ds  \mu_0 \,\int_{\R^3}   |{\bf H}(\bx)|^2  \, \md{\bf x} + \mu_0  \int_{\R^3}  \sum_\ell \big(\omega_{m,\ell}^2 \, | \bbM_\ell(\bx)|^2 +  |\dot \bbM_\ell(\bx)|^2\big) \, \md {\bf x}\,.
	\end{array} \right.
\end{equation}
 $\bbA$ and its domain are still defined by \eqref{eq.opA}  and \eqref{spaces}(i) with this time $\boldsymbol{\mathcal{D}_{s}} =\boldsymbol{\mathcal{D}_{s}}  = \boldsymbol{\mathcal{H}}_{e}$, and 
\begin{equation}\label{eq.opAbis}
\bbA := \ \rmi\, \begin{pmatrix}
	0 &\varepsilon_0^{-1}\,\mbox{\bf rot} & 0&-{\mathbb M}_e & 0 & 0\\[2pt]
- \mu_0^{-1}\,\mbox{\bf rot} & 0 &0 & 0& 0& -{\mathbb M}_m \\[2pt]
0 & 0 &0  & 1&0 &0\\[2pt]
{\mathbb M}_e^* & 0 &-\,\mbox{diag} \; \omega_{e,j}^2 &0 &0 &0\\[2pt]
0 & 0 &0  & 0 &0 & 1\\[2pt]
0 &{\mathbb M}_m^*& 0 & 0 & -\mbox{diag} \; \omega_{m,\ell}^2  & 0
\end{pmatrix},
\end{equation}
where $\mathbb{M}_s \in B(\boldsymbol{\cal H}_s, \bL^2(\R^3) )$ and its adjoint  $\mathbb{M}_s^*  \in B(\bL^2(\R^3) ,\boldsymbol{\cal H}_s), s = e,m$, are defined by 
$$
\forall \; (\mathbb{X}, {\bf F}) \in \boldsymbol{\cal H}_s \times  \bL^2(\mathbb{R}^3), \quad   \big({\mathbb M}_s \mathbb{X}\big)({\bf x}) := \sum_{q=1}^{N_s} \Omega_{s,q}^2 \, \mathbb{X}_q ({\bf x}), \quad \big({\mathbb M}_s^* {\bf F}(\bx)\big)_q = {\bf F}({\bf x}).
$$
\subsection{Dispersion analysis} \label{sec-disp-analysis}
This section is the only one for which the homogeneity of the medium, i. e. \eqref{homedium}, is required. We consider here local media (in the sense of  the definition of Example \ref{Exloc} where the permittivity  and permeability functions are rational functions of the forE
\begin{equation} \label{locallaw} \ds	\varepsilon(\omega) =   \frac{P_e( -\rmi \omega)}{Q_e( -\rmi \omega)} , \quad \mu (\omega) =   \frac{P_m(-\rmi \omega)}{Q_m(-\rmi \omega)},\end{equation} 
	were  $(P_e, Q_e, P_m, Q_m)$ are real polynomials with $d^oP_e = d^oQ_e$ and $d^oP_m = d^oQ_m$. We denote   by $\mathcal{P}_e$  (resp. $\mathcal{P}_m$)  the set of poles of the rational function $\varepsilon$ (resp. $\mu$) and by $\mathcal{Z}_e$  (resp. $\mathcal{Z}_m$) the corresponding set of zeros.
\\ [6pt]
The dispersion analysis of a homogeneous local medium  consists in looking at particular solutions under the form of plane waves, i. e. 
 in  the form
\begin{equation} \label{planewaves}
		{\bf E}(x,t) = \mathbb{E} \; \rme^{ \rmi ({\bf k} \cdot {\bf x} - \omega \; t)} , \quad 
		{\bf H}(x,t) = \mathbb{H} \; \rme^{\rmi ({\bf k} \cdot {\bf x} - \omega \; t)},
\end{equation}
where the wave vector is real, ${\bf k} \in \R^3$, and the frequency can be complex, $\omega \in \C$.  The constant vectors $(\mathbb{E}, \mathbb{H}) \in \R^3 \times \R^3$ describe how the electromagnetic field is {\it polarized}.
We assume here that $\omega\in \bbC\setminus (\calP_e\cup \calP_m  \cup \calZ_e\cup\calZ_m)$ and $\bk\neq 0$. Going back to Maxwell's equations \eqref{Maxwell} and the constitutive laws \eqref{Dispersive} and \eqref{locallaw}, one sees that \eqref{planewaves} provides a non trivial solution if and only if 
there exists $(\mathbb{D}, \mathbb{B}) \in \C^3 \times \C^3$  such that
\begin{equation} \label{Maxwell_freq}
\left\{\begin{array} {lll}
	\omega \, \mathbb{B}  + {\bf k} \times \mathbb{E}  = 0, & \quad \omega \, \mathbb{D}- {\bf k} \times \mathbb{E} = 0,  \\ [12pt]
	\mathbb{D} = \varepsilon(\omega) \, \mathbb{E}, & \quad \mathbb{B} = \mu(\omega) \, \mathbb{H}.
	\end{array} \right. 
\end{equation}
from which one deduces  after an  elimination of $(\mathbb{D}, \mathbb{B})$,  \eqref{planewaves}  is a  solution if and only if  
\begin{equation} \label{relationpolar}
{\bf k} \cdot \mathbb{E} = 0, \quad {\bf k} \cdot \mathbb{H} = 0, \quad \mathbb{E} \cdot \mathbb{H} = 0,
\end{equation}
and $(\omega, {\bf k})$ are related by the dispersion relation of the medium, namely 
\begin{equation} \label{relationdispersion}
{\cal D}(\omega) = |{\bf k}|^2, \quad \mbox{where } {\cal D}(\omega)  = \omega^2 \, \varepsilon(\omega) \mu(\omega),
\end{equation}
where $\varepsilon$ and $\mu$ are given by  \eqref{locallaw}.\\[6pt]
In the sequel, we shall make two assumptions on poles and zeros on the funtions $\varepsilon(\cdot)$ and $\mu(\cdot)$. 
\begin{Ass} \label{ireducible}
 $ \quad \displaystyle \calP_e \cap \calZ_m = \emptyset \mbox{ and }   \calP_m \cap \calZ_e= \emptyset. $
\end{Ass}
\begin{Ass} \label{ireduciblebis}
 $\quad 0\notin \calP_e\cup \calP_m.$
\end{Ass}
\begin{Rem}
The Assumption \ref{ireducible} is not restrictive in practise in  the sense explained in the section 3.2.1 of \cite{cas-kach-jol-17}. Avoiding  Assumption \ref{ireduciblebis} would require further developments which are not essential for the purpose of this chapter and would increase unnecessarily its length.
\end{Rem}
\noindent The above two assumptions mean that there is no possible simplification in the writing of the rational function $\cal{D}$  as a product $\omega^2 \, \varepsilon(\omega) \, \mu(\omega) $.
More precisely, Assumption \ref{ireducible} implies that the set of poles  of $\cal{D}$ is $\calP_e \cup \calP_m $ whereas, since Assumption  \ref{ireduciblebis} implies that $0$ is a zero with multiplicity at least $2$ of  $\cal{D}$,  the set of zeros  of $\cal{D}$ is $\calZ_e \cup \calZ_m  \cup \{ 0\}$.\\[6pt]
\noindent As a consequence, seen as an equation in $\omega$ parametrized by ${\bf k}$, the dispersion equation (\ref{relationdispersion}) can be seen as a polynomial equation of degree $N_{\mathcal{D}} =2 +d^oQ_e +d^oQ_m$ 
whose coefficients are affine functions  in $|{\bf k}|^2$ and whose higher order term is independent of $|{\bf k}|$. As a consequence, this equation admits $N$ continuous branches of solutions (see e.g. \cite{Hin-10}, Proposition 4.1.19): 
\begin{equation} \label{branches}
	|{\bf k}| \in \R^+\rightarrow \omega_j \big(|{\bf k}| \big) \in \C, \quad 1 \leq j \leq N_{\mathcal{D}},
\end{equation}
Moreover, the sign of the imaginary parts of these solutions is known: 
\begin{Thm} \label{disp_pass} In a passive local medium, the solutions $\omega_j \big(|{\bf k}|\big)$, see   \eqref{branches}, of the dispersion relation (\ref{relationdispersion}) satisfy $$\operatorname{Im}  \, \omega_j\big(|{\bf k}| \big)\leq 0, \quad \forall \; {\bf k} \in \R^3, \quad 1 \leq j \leq N_{\mathcal{D}}.$$ 
\end{Thm}
\begin{proof}
	By continuity, is suffices to prove the result for ${\bf k} \neq 0$. Assume that (\ref{relationdispersion}) admits for some solution $\omega$ with $\operatorname{Im}  \, \omega > 0$. Taking the real and imaginary part of (\ref{relationdispersion}), we thus get
	$$
	\begin{array}{l}
		(a) \quad		\operatorname{Re} \big(\omega \, \varepsilon(\omega)\big) \, \operatorname{Re} \big(\omega \, \mu(\omega)\big) = |{\bf k}|^2 +\operatorname{Im}  \big(\omega \, \varepsilon(\omega)\big) \, \operatorname{Im} \big(\omega \, \mu(\omega)\big) \\[12pt]
		(b) \quad		\operatorname{Re}\big(\omega \, \varepsilon(\omega)\big) \, \operatorname{Im}  \big(\omega \, \mu(\omega)\big) + \operatorname{Im}  \big(\omega \, \varepsilon(\omega)\big) \,  \operatorname{Re} \big(\omega \, \mu(\omega)\big) = 0
	\end{array}
	$$
By  passivity  and ${\bf (HF)}$,  $\omega \, \varepsilon(\omega)$ and $\omega \, \mu(\omega)$ are non constant Herglotz function, hence (by Remark  \ref{rem_herg})  we know   that $\operatorname{Im}  \big(\omega \, \varepsilon(\omega)\big) > 0$ and $\operatorname{Im} \big(\omega \, \mu(\omega))> 0$. Thus, we deduce from (a) that $\operatorname{Re} \big(\omega \, \varepsilon(\omega)\big)$ and $\operatorname{Re} \big(\omega \, \mu(\omega)\big)$ do not vanish and have the same sign. However,  (b) says that they are of opposite signs. This is a contradiction.
\end{proof}
\noindent When $\omega = \omega_R + \rmi \, \omega_I, (\omega_R,\omega_I) \in \R^2$,  plane wave solution (\ref{planewaves}) can be rewritten
\begin{equation} \label{planewave2}
	{\bf E}(x,t) = \mathbb{E} \; \rme^{ \rmi ({\bf k} \cdot {\bf x} -\omega_R \; t )}  \; e^{ \, \omega_I \, t}. 
\end{equation}
It corresponds to a wave propagating in the direction of the vector ${\bf k}$ at the phase~velocity~$\omega_R / |{\bf k}|$ with an amplitude which varies in time proportionally to $e^{\, \omega_I \, t}$. According to Theorem \ref{disp_pass}, $\omega_I \leq 0$, thus the amplitude decays in time, which corresponds to the intuitive notion of passivity. When $\omega_I = 0$, the amplitude is constant : the wave{\ purely propagative}. This justifies the
\begin{Def}{\bf (Non-dissipative media)}. A local medium is {non-dissipative} if and only˜if~all plane waves are purely propagative, i. e.~all solutions of the dispersion relation (\ref{relationdispersion}) are {real}.\end{Def}
\begin{Thm} \label{thm.disp_Lor} The Lorentz media \eqref{Lorentzlaws} are non dissipative.
\end{Thm}
\begin{proof} A specificity of Lorentz media is that the functions $\varepsilon(\omega)$ and $\mu(\omega)$ are even, thus the function ${\cal D}(\omega)$ (cf. \eqref{relationdispersion}) are even in $\omega$. Thus if $\omega$ a solution of \eqref{relationdispersion}, $- \omega$ is solution. Thus, if there would exist a solution with $\operatorname{Im}  \, \omega \neq 0$,  there would exist a solution with $\operatorname{Im}  \, \omega > 0$ which, according to Theorem \ref{disp_pass}, would contradict the passive nature of Lorentz media. \end{proof} 
\noindent  Lorentz media are somewhat universal  local models because  Theorem \ref{thm.disp_Lor} has a reciprocal.
\begin{Thm} \label{disp_Lor} Any passive local medium which is non dissipative is a Lorentz medium. 
\end{Thm}
\begin{proof}
First, as the medium is non-dissipative, the  solutions of the dispersion relation are real for any $\bk$, thus in particular for $\bk=0$. This means that the zeros of $\mathcal{D}$ are real. \\[4pt]
Now if $p \in \calP_e\cup \calP_m$ is a pole of $\mathcal{D}$  of multiplicity  $m > 0$. We can write
	$$
\omega^2 \, \varepsilon(\omega) \, \mu(\omega) \sim A_p\;\mathrm{e}^{ \rmi \, \theta_p} (\omega - p)^{-m}, \ \mbox{ for some } A_p>0, \;  \mbox{ and } \theta_p\in [0, 2\pi).
	$$ 
Rewriting the dispersion equation (\ref{relationdispersion}) as $A_p \, \mathrm{e}^{\rmi \, \theta_p} \; (\omega - p)^{-m} \big( 1 + O(\omega - p)\big) = |{\bf k}|^2$,  one deduces using the implicit function theorem (see lemma A.1 and proposition 4.1 of \cite{cas-jol-ros-22-bis} for  a rigorous justification) that  for $|{\bf k}|\rightarrow +\infty$, (\ref{relationdispersion}) admits $m$ branches of solutions $\omega_{p,\ell}(|{\bf k}|)$ given by
\begin{equation}\label{eq.pole}
\omega_{p,\ell}(|{\bf k}|) = p + A_{p}^{\frac{1}{m}} \; |{\bf k}|^{-\frac{2}{m}}\; \rme^{\rmi \, \big(\frac{2\ell \, \pi+\theta_p}{m} \big)}\ \; \big(1 + o(1) \big), \quad (|\bk| \to + \infty) \quad  \ell=1,\ldots, m.
\end{equation}
As $\omega_{p,\ell}(|{\bf k}|) \in \mathbb{R}$ and tends to $p$ when $|\bk| \to + \infty$,  $p\in \mathbb{R}$. Hence,  all the poles of $\mathcal{D}$  are real. \\[6pt]
 Thus, using ${\bf (HF)}$, and the factorized form of the numerator and denominator of $\omega \varepsilon(\cdot)$ (resp. $\omega \mu(\cdot)$), we deduce that $\omega \varepsilon(\cdot)$ (resp. $\omega \mu(\cdot)$)  is   real-valued on $\bbR \setminus \calP_e $ (resp. $\bbR \setminus \calP_m$).\\[6pt]
We prove now that the permittivity  $\varepsilon(\cdot)$ follows a Lorentz 	law  \eqref{Lorentzlaws} for some $N_e\geq 0$.  The proof is similar for the permeability  $\mu(\cdot)$.  As the medium is passive, $\omega \varepsilon(\cdot)$ is a Herglotz function which can be represented by   \eqref{expepsmu} (with $a=\varepsilon_0$ by $\bf (HF)$) where  its 
Borel  measure $\nu_e$ satisfies \eqref{propmes}   and can be computed via the identities  \eqref{measure}.
Using the fact that $\omega \varepsilon(\cdot)$  is real-valued on $\bbR \setminus \calP_e $,  one deduces from   \eqref{measure}
for any interval that contains no pole, i.e. $[a,b]\subset \bbR \setminus \calP_e $, $\nu_e\big([a,b]\big)=0$. Thus the support of $\nu$ is included in $\calP_e$ and $\nu_e$ is composed of a  finite number of Dirac masses  in \eqref{expepsmubis} which permits to conclude.
\end{proof}
\noindent We now analyse in more detail the properties of the dispersion relation \eqref{relationdispersion} for Lorentz media. \\ [12pt]
First of all, by evenness, the set of solutions of \eqref{relationdispersion} $({\bf k}$ being given) can be relabelled as 
\begin{equation} \label{solLorentz}
\left\{ 	\begin{array}{lll}
	\big\{ \pm \, \omega_{n}\big( |\bk| ), 1 \leq n \leq N+1,  N = N_e + N_m \big \}, \\ [12pt]
	\; 0 \leq \omega_{1}( |\bk|) \leq \omega_{2}( |\bk| ) \leq \cdots \leq  \omega_{N+1}( |\bk| ).
	\end{array} \right.
	\end{equation}
We introduce  the subsets of $\R^+$ attained by the functions $\omega_{n}\big( |\bk| \big)$.
\begin{equation} \label{defSn}  
{\cal S}_n := \mbox{closure } \big\{ \omega_{n}( |\bk| ), \bk \in \R ^3\}, \quad \mbox{ a closed sub interval of $\R^+$. }
\end{equation}
To describe our results, we introduce the set of positive poles and  zeros of the function ${\cal D}(\cdot)$, repeated with their multiplicity (which can be shown to be at most 2),
\begin{equation} \label{defZn} 
\left\{ 	\begin{array}{lll}	
	{\cal Z} :=\big\{ z \in \R^+ /  \, {\cal D}(\omega) = 0 \big\} \equiv \big\{ 0=z_1 < z_2 \leq \cdots \leq z_{N+1} \big \}, \\ [12pt]
	{\cal P} :=\big\{ p \in \R^+ /\,  {\cal D}(\omega)^{-1} = 0 \big\} \equiv \big\{ p_1 \leq p_2 \leq \cdots < p_{N+1} = + \infty\big\}, 
		\end{array} \right.
\end{equation}
~\\
where, as $\ds \lim_{\omega \rightarrow + \infty} {\cal D}(\omega) = + \infty$, we have considered that $+\infty$ is a pole of ${\cal D}(\omega) $.
\begin{Rem} Of course, 
${\cal P} = \{ \omega_{e,1}, \cdots, \omega_{e,N_e} \} \cup \{ \omega_{m,1}, \cdots, \omega_{m,N_m} \} $.
\end{Rem} 
\noindent As  a consequence of the fact the poles and  zeros  of $\omega \varepsilon(\cdot)$ (resp. $\omega \mu(\cdot)$) are simple and interlace along the real axis, see section \eqref{sec-Lorentz-def}, one deduces  that
\begin{equation}\label{eq.zero-pole-mult}
{\cal Z} \cup {\cal P} \mbox{ contains at most two consecutive poles or zeros counted with multiplicity}.
\end{equation}
We point out  that  even if,  by  Assumption \ref{ireduciblebis},  $z_1=0$  is a zero of multiplicity $2$ of $\mathcal{D}$,  it is  counted  only once  since  by evenness   we restrict our analysis  to  the positive  solutions  of \eqref{relationdispersion}. 
\begin{Thm} \label{thm_dispersion}  Under Assumptions \ref{ireducible} and \ref{ireduciblebis},  each function $|\bk|  \in \R^+ \mapsto \omega_{n}\big( |\bk| \big)$ is analytic on $ \R^{+,*}$ and strictly monotonous, the function $\omega_1(|\bk|)$ and $\omega_{N+1}\big( |\bk| \big)$ being strictly increasing, 
	\begin{equation} \label{limits_omega_n}
	(i)	 \quad  \omega_{n} (0)= z_n, \qquad  (ii) \quad \lim_{ |\bk|  \rightarrow + \infty} \omega_{n}\big( |\bk| \big) = p_n,
\end{equation}
		which means that  ${\cal S}_n = [z_n, p_n] \mbox{ or } [p_n, z_n]$. Moreover, the interiors of the intervals ${\cal S}_n$ are disjoint, in other words,
	\begin{equation} \label{disjoint}
\forall \; n \leq N, 	\quad \sup \, {\cal S}_n \leq \inf \, {\cal S}_{n+1}.
\end{equation}		
		\end{Thm}
\begin{proof}
The proof of this result is given in section 3.3 of \cite{cas-kach-jol-17}.  We recall it for the reader. \\[4pt]
{\bf Step 1: The rational dispersion relation \eqref{relationdispersion}  has only simple roots}.\\[4pt]
Let $\bk\neq 0$  and $\omega=\omega_n(|\bk|)$ for some $n\in \{ 1, \ldots, N+1\}$.  Thus $\omega \in \bbR^+ \setminus \calP$ satisfies  \eqref{relationdispersion}, i.e.
$$
\mathcal{D}(\omega)=(\omega\varepsilon)(\omega) \, (\omega \mu)(\omega)=|\bk|^2>0.$$ 
Hence, $\omega>0$ and   $ \varepsilon(\omega)$ and $ \mu(\omega)$ have the same sign and do not vanish. Furthermore, one has
\begin{equation}\label{sign}
\mathcal{D}'(\omega)=(\omega\, \varepsilon  \,  \omega\, \mu)'(\omega)=(\omega \, \varepsilon)'(\omega) \, \omega \mu(\omega)+\omega \, \varepsilon(\omega) \, (\omega \, \mu)'(\omega) .
\end{equation}
By Lemma \ref{lem_positivity}, 
$\omega \varepsilon$  and $\omega \mu$ satisfy the positivity property \eqref{positivity}: $(\omega \, \varepsilon)'(\omega) >0$ and $(\omega \, \mu)'(\omega) >0$. 
This implies (since $\omega>0$) that 
$\mathcal{D}'(\omega)\neq 0$ , i. e.  \eqref{relationdispersion}  has only  simple roots. Thus we can index them by increasing values as follows 
\begin{equation}\label{simplicity}
\omega_1(|\bk|)<\omega_2(|\bk|)<\ldots< \omega_{N+1}(|\bk|),  \quad \forall\; \bk\neq 0. 
\end{equation}
{\bf Step 2: Proof of the analyticity of the dispersion curves $|\bk|\mapsto \omega_{n}(|\bk|)$ on $\mathbb{R}^{+,*}$}.\\[4pt]
The function $G:(w,\zeta )\mapsto \mathcal{D}(w)-\zeta$ is analytic on a vicinity of $(\omega,|\bk|^2)$: indeed, it is analytic in $|\bk|$ and $\omega$ separately which allows us to conclude   by Hartog's theorem, see \cite{Muj-86}, Theorem 36.8 page 271. 
Thus, as $G( \omega_{n}(|\bk|), |\bk|^2)=0$ and $\partial_{w} G( \omega_{n}(|\bk|), |\bk|^2)= \mathcal{D}'( \omega_{n}(|\bk|)) \neq 0$ (by \eqref{sign}), 
 one  shows  easily via the analytic implicit function theorem (see e.g.  \cite{Fri-2002}, Theorem 7.6 page 34),  that the continuous function $|\bk| \to  \omega_{n}(|\bk|) $ is indeed  analytic on $\mathbb{R}^{+,*}$.\\[6pt]
{\bf Step 3: Proof of the strict monotonicity of $|\bk|\mapsto \omega_{n}(|\bk|)$.}\\[6pt]
As $\omega_n$ is analytic on $\bbR^{+,*}$, differentiating in $|\bk|$  the dispersion relation ${\cal D}\big((\omega_n(|{\bf k}|)\big)= |\bk|^2$ yields
\begin{equation} \label{identity}
\omega'_n \big(|{\bf k}| \big) \, {\cal D}' \big( \, \omega_n(|{\bf k}|)\big)= 2 \, |{\bf k}|>0.
\end{equation}
Hence, one has  for  all $|\bk|\neq 0$, $\omega'_n(|{\bf k}| )\neq 0$.  As $\omega'_n$ is  a continuous function, it has to keep the same sign for  all $|\bk|\neq 0$  given by the sign of $ \mathcal{D}'( \omega'_n(|{\bf k}|))$. 
Thus, using \eqref{sign} for $\omega = \omega_n(|{\bf k}| )$, $(\omega \, \varepsilon)'(\omega) >0$ and $(\omega \, \mu)'(\omega) >0$ and   the fact that $ \operatorname{sign}\varepsilon(\omega) = \operatorname{sign} \mu(\omega)$, we deduce that
\begin{equation}\label{eq.refsingbis}
 \operatorname{sign} \omega'_n \big(|{\bf k}| \big)=\operatorname{sign}\varepsilon\big( \omega_n(|{\bf k}| ) \big)=\operatorname{sign}\mu\big( \omega_n(|{\bf k}|)\big). 
\end{equation}
Thus, if $\varepsilon( \omega_n(|{\bf k}| )) $ and $\mu( \omega_n(|{\bf k}| ) )$ are both positive   (resp. both  negative), $ \omega_n$ is strictly increasing (resp. strictly decreasing) on $\bbR^+$.\\[6pt]
{\bf Step 4: Proof of the limits \eqref{limits_omega_n} }\\[6pt]
As $\omega_n(\cdot)$ is continuous at $0$,  $\omega_n(|\bk|)$ converges to $z_n=\omega_n(0)$ as $|\bk|\to 0$. This proves \eqref{limits_omega_n}(i). \\[6pt] 
 On the other hand, as $\omega_n$ is strictly monotonous, $\omega_n$  converge to a limit   $\ell_n$  (finite or infinite) when $|\bk|\to +\infty$, where, from \eqref{simplicity}, $\ell_1 \leq \ell_2 \leq  \cdots \leq\ell_{N+1}$.
Since $\omega_n(|\bk|)$ satisfies  \eqref{relationdispersion}, $\omega_n(|\bk|)^2\, \varepsilon(|\bk|) \mu(|\bk|) \to + \infty$ as $|\bk|\to +\infty$. Thus $\ell_n\in  \calP$, i.e. $\{\ell_1, \ell_2, \cdots, \ell_{N+1}\} \subset {\cal P}$.  It remains to prove that $\ell_n=p_n$. For this we observe that 
\begin{itemize} \item For $n \leq N$, i.e. $p_n <+\infty$, its multiplicity as a pole of $\mathcal{D}$ is $m = 1, 2$. Thus, as in the proof, 
 reasoning as in the proof of Theorem \ref{disp_Lor}, see \eqref{eq.pole}, one knows that there exists  for  $|\bk|$ large enough, $m$ distinct   branches   $\tilde \omega_j(\cdot)$ for $j\in \{ 1, m\}$ of solutions  such that  $$\tilde \omega_j(|\bk|)\to p_n, \quad \mbox{as } |\bk|\to +\infty .$$
\item  If $n=N+1$, since  by ${\bf (HF)} $,  $\omega^2 \varepsilon(\omega) \mu(\omega) \sim \omega^2 \varepsilon_0 \mu_0$ when $\omega\to+ \infty$, one expects the existence of one branch $\tilde \omega_{N+1}(\cdot)$ of solutions such that (the existence of this branch is proved rigorously via the implicit function, see lemma A.1 and proposition 4.1 of \cite{cas-jol-ros-22-bis})
$$\tilde \omega_{N+1}(|\bk|) = c \, |\bk| + o\big(|\bk|\big), \quad  |\bk|  \rightarrow + \infty, \ \  \mbox{  where $c=(\varepsilon_0\mu_0)^{\frac{1}{2}}$}.$$
\end{itemize} 
 As $\ell_1 \leq \ell_2 \leq  \cdots \leq\ell_{N+1}$,  the above shows by induction (left to the reader) that $\ell_n=p_n$ for $n=1,\ldots, N+1$. 
\\[6pt]
{\bf Step 5: Nature of the monotonicity of $\omega_1(\cdot)$ and  $\omega_{N+1}(\cdot)$.}\\[6pt]
By Assumption \ref{ireduciblebis}, $0\notin \calP_e\cup \calP_m $  which implies  that there is no Drude term (as defined in Remark \ref{LorentzDrude}) in $\varepsilon$ and $\mu$. Hence $\varepsilon(\omega)\sim \varepsilon(0)>0$ and  $\mu(\omega)\sim \mu(0)>0$ when $\omega \to 0^+$ (see \eqref{Lorentzlaws}), which implies,  by \eqref{eq.refsingbis}, 
that  the first band is strictly increasing.\\[6pt]
 As $p_{N+1}=+\infty$,  $\omega_{N+1}$ is strictly increasing and $S_{n+1}=[z_{N+1},+\infty[$.
\\[4pt]
{\bf Step 6: Proof of the non-overlapping property \eqref{disjoint}.}\\[4pt]
Assume by contradiction that $\sup \, {\cal S}_n > \inf \, {\cal S}_{n+1}$ for some $n\in \{1,\ldots, N\}$. Then,  the intervals $S_n$ and $S_{n+1}$ overlaps and there exists $k$ and $k'$ such that $\omega_{n}(|\bk|)=\omega_{n+1}(|\bk'|)$. Thus, the dispersion relation  \eqref{relationdispersion} implies that $|\bk|^2=\mathcal{D}(\omega_{n}(|\bk|))=\mathcal{D}(\omega_{n+1}(|\bk'|))=|\bk'|^2$.  Thus $|\bk|=|\bk'|$ and one concludes that $\omega_{n+1}(|\bk|)=\omega_{n}(|\bk|)$ which is impossible because of \eqref{simplicity}.
\end{proof}
\noindent		
When  $\omega = \omega_{n}( |\bk|)$, $\omega/ |\bk| > 0$ is by definition the {\it phase velocity} of the plane wave \eqref{planewaves}  while  $\omega'_{n}( |\bk| )$ is the corresponding 
while {\it group velocity} which, by Theorem \ref{thm_dispersion}, has a constant sign when $|\bk|$ varies. The corresponding dispersion curve $\Gamma_n$ as the (monotonous) graph:
\begin{equation} \label{def_Gamman}
	\Gamma_n := \big\{  \big(|\bk|, \omega_{n}( |\bk| )\big) , |\bk|  \in \R^+\}.
\end{equation}	
Theorem \ref{thm_dispersion}  implies,  for $\bk \neq 0$, one has the strict inequality   $\omega_{n}\big( |\bk| \big) < \omega_{n+1}\big( |\bk| \big)$ but, even more, the only possibility of non empty intersection between two dispersion curves occurs between $\Gamma_n$ and $\Gamma_{n+1}$ in the case of a double zero of ${\cal D}(\omega)$, $z_n = z_{n+1}$, in which case $\Gamma_n \cap \Gamma_{n+1}= \{z_n\})$.
\begin{equation} \label{defS}  
\hspace*{-5.5cm}\mbox{We define the set of {\it propagative} frequencies as } \quad 	{\cal S} := \bigcup_{n=1}^{N+1}  {\cal S} _n. 
\end{equation}
\begin{Lem} \label{lem_caracS} The set ${\cal S}$ is characterized by 
	\begin{equation} \label{charS}
		{\cal S} =  \operatorname{closure}  \big\{ \omega \in \R^+ \setminus {\cal P} \, \mid \, \varepsilon(\omega) \mu(\omega)>  0 \big \}. 
	\end{equation}	
\end{Lem} 
\begin{proof}
Let $\bk\in \bbR^3$. Then, by definition, $\omega_n(|\bk|)$ for  $n=1, \ldots, N+1$ are   the non negative solutions of the dispersion relation $\omega^2 \varepsilon(\omega) \mu(\omega) = |\bk|^2$. Thus, we have
\begin{equation}\label{eq.definterior}
  \big\{ \omega \in \R^+ \setminus {\cal P} \, \mid \, \varepsilon(\omega) \mu(\omega)>  0 \big \}=\big\{ \omega_n(|\bk|), \ \bk \in \bbR^3\setminus \{ 0\} \mbox{ and } n\in \{1,\ldots, N\}  \big\}.
  \end{equation}
By virtue of the definition of $\mathcal{S}$ (see \eqref{defS} and  \eqref{defSn}) and the continuity of the function $\omega_n$ at $0$, one obtains  \eqref{charS} by taking the closure of the above relation.
\end{proof}
\noindent 
	It can be shown, see (Theorem 4.16, remark 4.18 of \cite{cas-jol-ros-22-bis}), that ${\cal S}$ is related to the spectrum $\sigma(\bbA)$ of the operator $\bbA$ appearing in the abstract formulation of section \ref{sec-abstract2}, see \eqref{eq.opAbis}, by
	\begin{equation*} \label{defSmS}  
	\sigma(\bbA) =	{\cal S} \cup \big(- {\cal S} \big) .
	\end{equation*}
For this reason, the sets ${\cal S} _n$ are called spectral bands.
	By opposition, we define the set of {\it non propagative} or {\it evanescent} frequencies as 
	\begin{equation} \label{defG}  
	{\cal G} := \R^+ \setminus 	{\cal S} .
	\end{equation}
	By \eqref{defS}, ${\cal G}$ is  a finite union ($\leq N$) of open bounded intervals,  called {\it spectral gaps}.\\ [12pt] 
	According to theorem \ref{thm_dispersion}, we have
	$$
	[1, \cdots ,N+1] = {\cal N}_+ \cup {\cal N}_-, \quad {\cal N}_+ := \big\{ n \, \mid \, \omega_n'> 0 \mbox{ on } \bR^{+,
*} \big\}, \quad {\cal N}_- := \big\{ n \, \mid \, \omega_n'< 0 \mbox{ on } \bR^{+,
*} \big\}, 
	$$ 
	which allows us to decompose ${\cal S}$ as 
\begin{equation} \label{decompo_spectre}
		{\cal S} = {\cal S}_+ \cup {\cal S}_-, \quad {\cal S}_+ := \bigcup_{n \in {\cal N}_+} {\cal S}_n, \quad {\cal S}_- := \bigcup_{n \in {\cal N}_-} {\cal S}_n
\end{equation}
		where by definition 
${\cal S}_+$ is the set of {\it positive} (or {\it forward}) frequencies and  ${\cal S}_-$ is the set of {\it negative} (or {\it backward}) frequencies: the terminology {\it  positive/forward} or  {\it negative/backward} refers to the sign on the group velocity. Accordingly, the ${\cal S}_n$ for $n \in {\cal N}_+$ are called positive spectral bands and the ${\cal S}_n$ for $n \in {\cal N}_-$ are called negative spectral bands.\\ [12pt]
Note that the last band, which is unbounded,  is always positive. Due to Assumption \ref{ireduciblebis}, the  first band is also positive since, as $z_1=0$, it is of the form $[0, p_1]$ .
\begin{Def} [Negative index material]
A material will be called {\it negative} material (or {\it negative index} material) if ${\cal S}_-$ is not empty, i.e. if there exists {\it at least} one negative spectral band.	
\end{Def}
\noindent According to Theorem \ref{thm_dispersion}, the fact that a material is negative corresponds to the existence of $n \in [1,\cdots, N]$ such that $z_n > p_n$. In fact, this property can be easily characterized in terms of  the functions $\varepsilon(\omega)$ and $\mu(\omega)$ along the real axis. 
\begin{Thm} \label{thm_dispersionbis} The sets ${\cal S}^\pm$ are characterized by 
\begin{equation} \label{charSpm}
	\left\{	\begin{array}{lll}
{\cal S}_+ =  \operatorname{closure}   \big\{ \omega \in \R^+ \setminus {\cal P} \, \mid \, \varepsilon(\omega) >  0 \mbox{ and } \mu(\omega)> 0 \big \},
\\[12pt]
{\cal S}_- =  \operatorname{closure} \big\{ \omega \in \R^+ \setminus {\cal P} \, \mid \, \varepsilon(\omega) < 0 \mbox{ and } \mu(\omega) < 0 \big \}.
\end{array} \right.
	\end{equation}
In particular, a Lorentz material is a negative material if and only if there exists a non empty open sub-interval of $\R^+$ along which $\varepsilon(\omega) > 0 \mbox{ and } \mu(\omega)$ are negative together. 
\end{Thm}
\begin{proof}
Let $n\in \{ 1, \ldots, N+1\}$ be fixed. For $\bk\neq 0$, $\omega_n(|\bk|)$ is a  positive solution of the dispersion relation \eqref{relationdispersion}. As $\omega_n$ is continuous (and even analytic) function on $(0, \infty)$, $\varepsilon(\omega_n(|\bk|))$ and $\mu(\omega_n(|\bk|))$ does not vanish on $(0,\infty)$ and have the same sign. As it has been point out in   \eqref{eq.refsingbis}, this  sign determines the sign of $\omega'_n(|\bk|)$ which is thus also constant on $(0,\infty)$. So that, by virtue of  \eqref{eq.definterior}, one gets that
$$
\{ \omega \in \R^+ \setminus {\cal P} \, \mid \,  \pm  \, \varepsilon(\omega) >  0 \mbox{ and } \pm \,  \mu(\omega)> 0 \big \}=\big\{ \omega_n(|\bk|), \ \bk \in \bbR^3\setminus \{ 0\} \mbox{ and } n\in \{1,\ldots,  {\cal N}_\pm  \}  \big\}
$$
By virtue of the definition of $\mathcal{S}_{\pm}$ (see \eqref{decompo_spectre} and  \eqref{defSn}) and the continuity of the function $\omega_n$ at $0$, one obtains  \eqref{charSpm} by taking the closure of the above relation.
\end{proof}
 \noindent We summarize the results in section 2.5 by the figure \ref{Fig_Disp}.
	\begin{figure}[h!] 
	\centerline{
		\includegraphics[width=0.6 \textwidth]{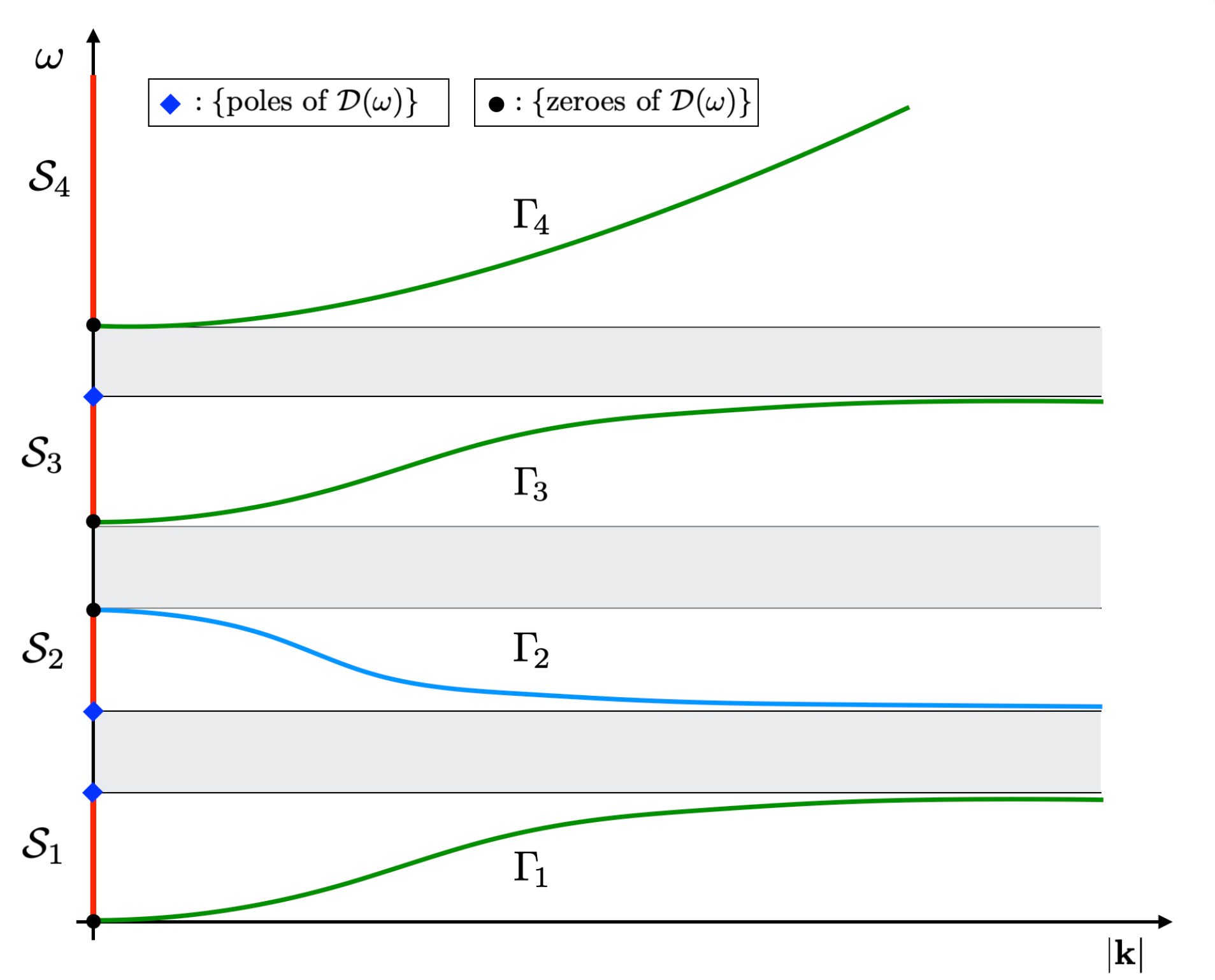}}
	\caption{Plot of dispersion curves for $N=3$. The spectral bands ${\cal S}_n$ are the projections of the dispersion curves $\Gamma_n$ on the $\omega$-axis. There is one negative band: ${\cal S}_2$.  }
	\label{Fig_Disp}
\end{figure}
\begin{Example}[Drude model] The Drude model is the simplest Lorentz model, corresponding to $N_e = N_m = 1$ and $\omega_{e,1} = \omega_{e,1} = 0$. This leads to 
(with $\Omega_{e,1} \equiv \Omega_{e}$ and $\Omega_{m,1} \equiv \Omega_{m}$)
\begin{equation} \label{Drudelaws}
	\varepsilon(\omega)	= \varepsilon_0 \; \Big( 1 - \frac{\Omega_{e}^2}{\omega^2}\Big), \quad \mu(\omega)	= \mu_0 \; \Big( 1 - \frac{\Omega_{m}^2}{\omega^2}\Big).
\end{equation}
Let us assume, without loss of generality, that $\Omega_e \leq \Omega_m$. As $\varepsilon(\omega)$ and $\mu(\omega)$ are negative for $0 \leq \omega < \Omega_e$, the Drude material is a negative material . Moreover, when $\Omega_e <\Omega_m$, one can check that the medium has one spectral gap, namely ${\cal G} = ]\Omega_e, \Omega_m[$, and two spectral bands:
$$
{\cal S}_1 = [0, \Omega_e] \mbox{ (a negative spectral band)}, \quad {\cal S}_2 = [\Omega_m, + \infty] \mbox{ (a positive spectral band)}, \quad 
$$
In the {\it critical} limit case  $\Omega_e =  \Omega_m$, the gap disappears and the two spectral bands touch.   The fact that ${\cal S}_1$  is a negative band seems  to  contradict Theorem \ref{thm_dispersion}. Indeed, there is no contradiction, since for the Drude material   $0\in \calP_e\cap \calP_m$. This means that Assumption \ref{ireduciblebis}, one of the assumptions of Theorem \ref{thm_dispersion},  does not hold for this material.
\end{Example}

\section{Dissipative Lorentz models} \label{Lorentz-dissp}
\subsection{Dielectric permittivity and magnetic permeabilty} \label{Lorentz-dissp-permittivity-permeability}
Dissipative Lorentz media are deduced from the non dissipative Lorentz laws \eqref{Lorentzlaws} by adding some imaginary part  to the denominators in the fraction expansion of $\varepsilon(\omega)	$ and $ \mu(\omega)$. This involves  (electric and magnetic) damping coefficients $\alpha_{e,j} \geq 0$ and  $\alpha_{m,\ell} \geq 0$. We shall note 
\begin{equation} \label{Damping coeff}
	\boldsymbol{\alpha}_e := ( \alpha_{e,j}) \in \R^{N_e}, \quad \boldsymbol{\alpha}_m:= ( \alpha_{m,\ell}) \in \R^{N_m}, \quad \boldsymbol{\alpha} := (\boldsymbol{\alpha}_e, \boldsymbol{\alpha}_m) \in \R^N.
\end{equation}
The corresponding permittivity and permeability, that we choose to index by $\boldsymbol{\alpha}$ to distinguish them form \eqref{Lorentzlaws} in the non dissipative case, are then given by:
\begin{equation} \label{DissipativeLorentzlaws}
	\varepsilon_\bal(\omega)	= \varepsilon_0 \; \Big( 1 + \sum_{j = 1}^{N_e} \frac{\Omega_{e,j}^2}{\omega_{e,j}^2 - \rmi \, \alpha_{e,j} \, \omega -\omega^2}\Big), \quad \mu_\bal(\omega)	= \mu_0 \; \Big( 1 + \sum_{\ell = 1}^{N_m} \frac{\Omega_{m,\ell}^2}{\omega_{m,\ell}^2 - \rmi \, \alpha_{m,\ell} \, \omega -\omega^2}\Big).
\end{equation}
In this section we shall assume that  $\boldsymbol{\alpha} \neq 0$ i.e. that at least one damping coefficient is positive. 
\noindent	The fact that $\omega \, \varepsilon_\bal(\omega)$  and $\omega \, \mu_\bal(\omega)$ are Herglotz functions can be checked directly by hand. \\ [12pt]By the way, one can make explicit the corresponding Herglotz measures $\md \nu_e(\xi)$ and $\md \nu_m(\xi)$, cf. \eqref{expepsmu},  which, contrary to the non dissipative case, are no longer pure point measures associated to resonances as in the non dissipative case, but the sum of Dirac measures and an absolutely continuous measure with respect to Lebesgue measure. More precisely, 
one can check that 
\begin{equation} \label{measuresbis}
	\left\{\begin{array}{llclc}
\ds 	\md \nu_e(\xi) & = & \ds \sum_{\alpha_{e,j}=0} \frac{ \Omega_{e,j}^2 }{2} (\delta_{\omega_{e,j}}  +  \delta_{-\omega_{e,j}} ) & + & \ds \sum_{{\alpha_{e,j}\neq0}}  \Omega_{e,j}^2 \, \nu_{e,j}(\xi) \, \md \xi, \\ [24pt] 
\ds	\md \nu_m(\xi)& = & \ds \sum_{\alpha_{m,\ell}=0} \frac{ \Omega_{m,\ell}^2 }{2}(\delta_{\omega_{m,\ell}} + \delta_{-\omega_{m,\ell}} )  & + &  \ds \sum_{\alpha_{m,\ell}=0}  \Omega_{m,\ell}^2 \, \nu_{m,\ell}(\xi) \, \md \xi,
	\end{array} \right.
\end{equation}
where the densities $\nu_{e,j}(\xi)$ and $\nu_{m,\ell}(\xi)$ are given by 
\begin{equation} \label{densities}
 \nu_{e,j}(\xi) = \frac{1}{\pi} \, \frac{\alpha_{e,j} \,  \xi^2}{(\xi^2 -\omega_{e,j}^2 )^2 +  \, \alpha_{e,j}^2 \, \xi^2} , \quad  \nu_{m,\ell}(\xi) = \frac{1}{\pi} \, \frac{\alpha_{m,\ell} \,  \xi^2}{(\xi^2 -\omega_{m,\ell}^2 )^2 +  \, \alpha_{m,\ell}^2 \, \xi^2}
\end{equation}
\begin{Rem} It is an exercise for the reader to recover (\ref{measuresbis}, \ref{densities}) from the general theory (see \eqref{measure}). It is also an exercise to check that, in the sense of measures, one has  the following weak convergence:
$$
\nu_{e,j}(\xi) \, \md \xi \; \longrightarrow \; \frac{1}{2}\big(\delta_{\omega_{e,j}} + \delta_{-\omega_{e,j}}\big) \  (\alpha_{e,j} \rightarrow 0), \quad \nu_{m,\ell}(\xi) \, \md \xi \; \longrightarrow \; \frac{1}{2}\big(\delta_{\omega_{m,\ell}} + \delta_{-\omega_{m,\ell}}\big)  \ ( \alpha_{m,\ell}\rightarrow 0),
$$
\end{Rem} 
\noindent In the spirit of section \ref{Augmented}, one would use the (non local) conservative augmented formulation \eqref{Lorentzsystemgene}. However, since this formulation corresponds to the conservation of an energy, namely ${\cal E}_{cons}(t)$, it is not adapted to put in evidence the dissipative nature of the material and proving the long time decay of the  electromagnetic energy, which will be the main objective on this section. 
That is why we are going to use an alternative augmented formulation. 
\subsection{Dissipative augmented formulation}\label{sec-disspative-evol}
\noindent 
This formutation is somewhat more natural than \eqref{Lorentzsystemgene}: it is directly issued from the rational form \eqref{DissipativeLorentzlaws} of $	\varepsilon_\bal(\omega)$ and $ \mu_\bal(\omega)$. This leads to the following system of equations where each additional unknown, ${\bf P}_j $ or ${\bf M}_\ell$, is associated in each of the denominators in \eqref{DissipativeLorentzlaws}:
	\begin{equation} \label{DissipativeLorentzsystem}
	\left\{	\begin{array}{lll}
		\ds	\varepsilon_0 \, \partial_t {\bf E} + {\bf rot} \, {\bf H} + \varepsilon_0 \, \sum_{j=1}^{N_e} \Omega_{e, j}^2  \, \partial_t \bbP_j= 0, & \quad  \partial_t^2 \bbP_\ell  + \alpha_{e,j} \, \partial_t \bbP_j + \omega_{e,j}^2 \,  \bbP_j = \,  {\bf E},\\[18pt]
		\ds \mu_0 \, \partial_t {\bf H} - {\bf rot} \, {\bf E}  + \mu_0 \,  \sum_{\ell = 1}^{N_m} \Omega_{m, \ell}^2 \, \partial_t \bbM_\ell = 0, 	 &  \quad \partial_t^2 \bbM_\ell + \alpha_{\ell,m} \, \partial_t \bbM_\ell + \omega_{m, \ell}^2 \, \bbM_\ell =  \,  {\bf H}.
	\end{array} \right.
\end{equation}
This system is completed by the divergence free initial conditions
\begin{equation} \label{CI}
	\left\{\begin{array}{l}
		\mathbf{E}(\cdot, 0) =  \mathbf{E}_0, \  \mathbf{H}(\cdot, 0) =  \mathbf{H}_0, \ \bbP(\cdot, 0)= \ \partial_t \bbP(\cdot, 0) =0,   \ \bbM(\cdot, 0)= \ \partial_t \bbM(\cdot, 0) =0  \\ [8pt] 
		\mbox{ with }  \ \mathbf{E}_0, \,  \bH_0 \in \bL^2(\mathbb{R}^3)  \ \mbox{ and } \  \nabla \cdot \mathbf{E}_0=\nabla \cdot \mathbf{H}_0=0,
\end{array} \right.
\end{equation}
where  $\bL^2(\mathbb{R}^3)=L^2(\mathbb{R}^3)^3$ and
\begin{equation*} \label{notPM} \bbP=(\bbP_j)\quad \mbox{and} \quad {\bbM}=(\bbM_\ell) \quad \mbox{ with } \quad  (\bbP_j):=(\mathbf{P}_j)_{j=1}^{N_e} \ \mbox{ and } \ (\bbM_\ell):=(\bbM_\ell)_{\ell=1}^{N_m}.
\end{equation*}
The reader will notice that one can assume without any loss of generality that the couples $(\alpha_{e,j}, \omega_{e,j})$ (resp. $(\alpha_{m,\ell}, \omega_{m,\ell})$) are all distinct the ones from the others.\\[6pt] 
\noindent Note that the above dissipative Lorentz models are purely local and differ from non dissipative Lorentz models \eqref{Lorentzsystem} only by adding a first order term in the ODE's for the auxiliary unknowns.
The total energy  ${\cal E}_{dis}(t)={\cal E}(t)+ {\cal E}_{e}(t)+{\cal E}_{m}(t)$ of the solution of the Cauchy problem associated to \eqref{DissipativeLorentzsystem}, still defined  as ${\cal E}_{cons}(t)$  by (\ref{defEnergy}, \ref{additionalenergiesbis}),  is this time no longer conserved but is a decreasing function of time, thanks to the following energy identity (whose proof follows the one of \eqref{thm.Energygene}):
	\begin{equation} \label{Dissipationenergy}
		\frac{d}{dt} {\cal E}_{dis}(t) + \Delta_{\boldsymbol{\alpha}}(t) = 0, \quad \Delta_{\boldsymbol{\alpha}}(t)  = \Delta_{\boldsymbol{\alpha},e}(t)  + \Delta_{\boldsymbol{\alpha},m}(t),
		\end{equation}
		where the positive dissipation functions $\Delta_{\boldsymbol{\alpha},e}(t)$ and $\Delta_{\boldsymbol{\alpha},m}(t)$ are given by 
	\begin{equation} \label{dissipationfunctions}
	\left\{	\begin{array}{l}
		\ds	\Delta_{\boldsymbol{\alpha},e}(t) :=   \varepsilon_0 \ \sum_{j=1}^{N_e} \alpha_{e,j}\, \Omega_{e, j}^2 \int_{\R^3} |\partial_t {\bf P}_{\!j}({\bf x}, t)|^2  \; d{\bf x},\\[15pt]
		\ds  
		\ds	\Delta_{\boldsymbol{\alpha},m}(t) :=  \mu_0 \ \sum_{\ell=1}^{N_m} \alpha_ {\ell,m}\Omega_{m, \ell}^2 \int_{\R^3}  |\partial_t {\bf M}_\ell({\bf x}, t)|^2   \; d{\bf x}.
	\end{array} \right.
\end{equation}
The identity  \eqref{Dissipationenergy} expresses a dissipation result from which one can expect that the total energy ${\cal E}_{dis}(t)$, thus a fortiori the electromagnetic energy ${\cal E}(t)$, tends to $0$ when $t \rightarrow + \infty$. This question was the object of the two papers \cite{cas-jol-ros-22,cas-jol-ros-22-bis}, in which we analyzed the rate of decay  of ${\cal E}(t)$ to $0$ for large time.  In the following, we wish to report on the results of the second paper, based on a modal analysis via the use of the Fourier transform in space. There results deeply relies on a sharp analysis of the dispersion relation \eqref{relationdispersion} for the dissipative Lorentz model \eqref{DissipativeLorentzlaws}, which is the object of the section \ref{sec_DispersionDissipative}.  This dispersion analysis results into decay estimates for the energy ${\cal E}(t)$, as we shall see in sections \ref{sec_decay}  and \ref{sec_decay_proof}.
\subsection{Dissipative dispersion relation} \label{sec_DispersionDissipative}
The dispersion relation of \eqref{DissipativeLorentzsystem} writes
\begin{equation} \label{relationdispersionDiss}
	{\cal D}_\bal(\omega) = |{\bf k}|^2, \quad \mbox{where } {\cal D_\bal}(\omega)  :=  \omega^2 \, \varepsilon_\bal(\omega) \mu_\bal(\omega).
\end{equation}
Contrary to ${\cal D}(\omega)$ in the non dissipative case, ${\cal D}_\bal(\omega)$ is no longer even but satisfies the symmetry property $\overline{{\cal D}_\bal(- \overline \omega) }= {\cal D}_\bal(\omega)$ (which corresponds to $({\bf RP})$). \\[6pt]
We introduce the second order polynomials  $q_{e,j}$ and $q_{m,\ell}$ are defined by
\begin{equation}\label{eq.polynom}
	q_{e,j}(\omega)=\omega^2+ \rmi\, \alpha_{e,j} \, \omega- \omega_{e, j}^2 \  \mbox{ and } \ q_{m,\ell}(\omega)=\omega^2+ \rmi \, \alpha_{m, \ell} \, \omega- \omega_{m, \ell}^2 .
\end{equation}
\noindent We make now an additional assumption for the irreducibility of the dispersion relation \eqref{relationdispersionDiss}:
\begin{Ass} \label{ireducible2}
Electric polynomials  $q_{e,j}$ (see \eqref{eq.polynom})  with  distinct indices $j$  do not have common roots. The same holds for the magnetic  polynomials $q_{m,\ell}$  with  distinct indices $\ell$. 
\end{Ass}

\begin{Rem}
When $\alpha_{e,j} <  2 \, \omega_{e,j}$, the two roots of $q_{e,j}$, $\omega_*$ and $-\overline{\omega_*}\notin \rmi \R^-$, are distinct.
Moreover, as  $\alpha_{e,j}=-2\operatorname{Im}(\omega_*)$ and $\omega_{e,j}= |\omega_*|$, $q_{e,j}$  can not share a common root with an other electric polynomial $q_{e,j'}$ since by assumption: $(\alpha_{e,j},\omega_{e,j})\neq (\alpha_{e,j'},\omega_{e,j'})$ for $j\neq j'$. Therefore, one only needs to assume $(\mathrm{H}_1)$ for electric polynomials $q_{e,j'}$ for which $\alpha_{e,j'}\geq 2 \, \omega_{e,j'} $. The same properties hold for the magnetic polynomials $q_{m,\ell}$ with obvious changes.  
\end{Rem}
%

\noindent It is easy to check that  with the Assumptions \ref{ireducible} and \ref{ireducible2} the rational functional function  ${\cal D}_\bal$ is irreducible. Moreover, its numerator is of degree $2 N+2$ and its denominator is of degree $2N$ (where we recall that  $N=N_e+N_m$). \\[6pt]
As in the non dissipative case, the set of poles ${\cal P}_\bal $  and zeros  ${\cal Z}_\bal $ of the function ${\cal D}_\bal$ will play a particular role. Due to the symmetry property,  thes sets are invariant by the transformation $\omega \to - \overline{\omega}$ and more precisely if $z_{\alpha}$ (resp. $p_{\alpha}$ ) is a zero (resp. a pole) of multiplicity $m$ of a ${\cal D}_\bal$ if and only if 
$-\overline{z_{\alpha}}$ (resp. $-\overline{p_{\alpha}}$ ) is a zero (resp. a pole) of multiplicity $m$ of ${\cal D}_\bal$. 
\begin{Def} \label{Resonances} $[$ Resonances/Non dissipativity $] $We shall say that 
	$\pm \omega_{e,j}$ is an electric resonance if $\alpha_{e,j}=0$, 
	and $\pm \omega_{m,\ell}$ is an magnetic resonance if  $\alpha_{m,\ell}=0$.
	\noindent We say that the medium is electrically non dissipative (e.n.d)  if all $\alpha_{e,j}$vanish, and that 
	it is  magnetically non dissipative (m.n.d) if all $\alpha_{m,\ell}$ vanish. We denote $\text{R}_{\text e} $, resp. $\text{R}_{\text m} $, the set  of electric, resp. magnetic, resonances. Of course, $\text{R}_{\text e} $ is the set of real poles of $\omega \, \varepsilon(\omega)$,  $\text{R}_{\text m} $ is the set of real poles of $\omega \, \mu(\omega)$ and  $\text{R}_{\text e} \cup \text{R}_{\text m}$ is the set of real poles of ${\cal D}_\bal$. 
\end{Def}
\noindent We split the set of electric resonances into two disjoint subsets $\text{R}_{\text e} = \text{R}^s_{\text e}  \cup \text{R}^d_{\text e}$ where  
\begin{equation} \label{defRe} 
	\pm \omega_{e,j} \in \text{R}_{\text e}^s  \quad \Longleftrightarrow \quad \omega_{e,j} \in \text{R}_{\text e} \mbox{ and } \omega_{e,j} \notin \{\omega_{m,\ell}\}.
\end{equation}
In the same way, $\text{R}_{\text m} = \text{R}^s_{\text m}  \cup \text{R}^d_{\text m}$ where  
\begin{equation} \label{defRm} 
	\pm \omega_{m,\ell} \in \text{R}_{\text m}^s  \quad \Longleftrightarrow \quad \omega_{m,\ell} \in \text{R}_{\text m} \mbox{ and } \omega_{e,j} \notin \{\omega_{e,j}\}.
\end{equation}
Note that $\pm \omega_{e,j} \in \text{R}_{\text e}^s$ iff  $\pm \omega_{e,j} \in \text{R}_{\text e}$ and is a simple pole of ${\cal D}_\bal$, while   $\pm \omega_{m,\ell} \in \text{R}_{\text m}^s$ iff   $\pm \omega_{m,\ell} \in \text{R}_{\text m}$  and is a double pole of ${\cal D}_\bal$ (this explains the the indices ${s}$ and ${d}$). Also note that $  \text{R}_{\text e}^d= \text{R}_{\text m}^d$.
\begin{Def} \label{Critical_cases} $[$Weak/strong dissipativity $]$
	Considering the Lorentz model \eqref{DissipativeLorentzlaws}, we say that
	\begin{itemize}
		\item[(i)]  the medium is magnetically weakly dissipative if it is e.n.d  and $\text{R}_{\text e}^s\neq \emptyset$. 
		\item[(ii)]    the medium is electrically weakly dissipative if it is m.n.d and $\text{R}_{\text m}^s \neq \emptyset$. 
	\end{itemize}
	\noindent	The medium is weakly dissipative  if it is magnetically or electrically weakly dissipative. Otherwise, we shall say that the medium is strongly dissipative. 
\end{Def} 
\noindent Note that, of course, dissipative Lorentz media are generically  strongly dissipative.

\begin{Thm} \label{thm_dispersion_diss} Seen as an equation in $\omega$, the dispersion relation \eqref{relationdispersionDiss} admits $2N+2$ solutions $$\big \{\omega_{n}^\bal( |\bk|) , 1 \leq n \leq  2N+2\big\} $$
where each function $|\bk|  \in \R^+ \mapsto \omega_{n}^\bal( |\bk|)$ is  continuous and piecewise analytic  
and satisfies 
\begin{equation} \label{prop_dissip}
	 \forall \; \bk \neq 0, \quad \operatorname{Im} \, \omega_{n}^\bal\big( |\bk| \big) < 0.
	 \end{equation} 
	 Moreover, one has $\omega_{n}^\bal (0)= z_n^\bal \in \calZ_{\bal}, \ \forall  n\in  \{1, \ldots , 2N+2\}$ and the numeration of the functions $\omega_n^\bal$ is chosen such that 
\begin{equation} \label{limits_omega_n_alpha}
 \lim_{ |\bk|  \rightarrow + \infty} \omega_{n}^\bal\big( |\bk| \big) = p_n^\bal \in \calP_{\bal} \mbox{ if }  1\leq  n \leq 2N \mbox { and }   \ds   \lim_{ |\bk|  \rightarrow + \infty} \, | \, \omega_{n}^\bal( |\bk| )| = + \infty \mbox{ if  } n \geq 2N+1. 
\end{equation}
Furthermore, if $z^\bal\in   \calZ_{\bal}$ is a zero (resp. $p^\bal \in  \calP_{\bal}$ is  a pole) of multiplicity $m$  of ${\cal D}_\bal$, there exists exactly $m$ distinct branches  $\omega_n^\bal $ which converge to  $z^\bal$ as $|\bk|\to 0$ (resp. to $p^{\bal}$ as $|\bk|\to +\infty$).
Moreover, if $|\bk|$ is large enough   or  $|\bk|>0$ is small enough, the  branches of solutions $\omega_n^{\bal}$ do not cross each other.
	\end{Thm}
	\begin{proof}
Let $\bal$ be fixed.  The rational dispersion relation \eqref{relationdispersionDiss} is equivalent to a polynomial equation $P^{\bal}_ {|\bk|}(\omega)=0$  (parametrized by $|\bk|$) where  $P^{\bal}_{|\bk|}$  is a polynomial  of degree $2N+2$  whose leading  coefficient  $\varepsilon_0 \, \mu_0>0$  is independent of $|\bk|$. 
Thus, the existence of $N$ continuous   and piecewise analytic functions $\omega_n^\bal$ for $n\in \{ 1, \ldots, 2N+2\}$ is  a consequence of Theorem 4.1.14  and Proposition 4.1.19  of \cite{Hen-86}.\\[6pt]
\noindent The limits \eqref{limits_omega_n_alpha} and the asymptotic properties of $ \omega_{n}^\bal$  for $|\bk|$ large enough or $|\bk|>0$  small enough  are obtained   as in the non-dissipative case by  the implicit function theorems, see Propositions 4.1 and 5.1  of \cite{cas-jol-ros-22-bis} for the technical details. \\[6pt]
	Concerning \eqref{prop_dissip}, we already know by Theorem \ref{disp_pass} that all solutions $\omega$ of  \eqref{relationdispersionDiss} for ${\bf k} \neq 0$ satisfy $\operatorname{Im}  \, \omega \leq 0$. It remains to eliminate the possibility of  a real solution $\omega$. If by contradiction, such a  real solution $\omega$ exists, it cannot be a magnetic or electric resonance since these are points where ${\cal D}_\bal(\omega)$ blows up.  For any other real  solution of \eqref{relationdispersionDiss}, we have, as in  the proof of Theorem \ref{disp_pass} 
		$$
		\begin{array}{l}
			(a) \quad		\operatorname{Re}\big(\omega \, \varepsilon_\bal(\omega)\big) \, \operatorname{Re} \big(\omega \, \mu_\bal(\omega)\big) = |{\bf k}|^2 + \operatorname{Im}  \big(\omega_\bal \, \varepsilon_\bal(\omega)\big) \, \operatorname{Im} \big(\omega \, \mu_\bal(\omega)\big) \\[12pt]
			(b) \quad		\operatorname{Re} \big(\omega \, \varepsilon_\bal(\omega)\big) \, \operatorname{Im}  \big(\omega \, \mu_\bal(\omega)\big) + \operatorname{Im}  \big(\omega \, \varepsilon_\bal(\omega)\big) \, \operatorname{Re} \big(\omega \, \mu_\bal(\omega)\big) = 0
		\end{array}
		$$
 As $\omega\in \bbR$, one has  $\operatorname{Im}  \big(\omega_\bal \, \varepsilon_\bal(\omega)\big) \geq 0$ and $\operatorname{Im}  \big(\omega_\bal \, \mu_\bal(\omega)\big) \geq 0$ and one deduces from (a) that $\operatorname{Re} \big(\omega \, \varepsilon_\bal(\omega)\big)$ and $\operatorname{Re} \big(\omega \, \mu_\bal(\omega)\big)$ have the same sign (and do not vanish). To obtain a contradiction,  as for Theorem \ref{disp_pass}, we simply have to check that  one has   $\operatorname{Im}  \big(\omega \, \varepsilon_\bal(\omega)\big)$ or $\operatorname{Im}  \big(\omega \, \mu(\omega)\big) > 0$. Indeed, from  (b), $\operatorname{Im}  \big(\omega \, \varepsilon_\bal(\omega)\big) > 0$ 
 implies $\operatorname{Im}  \big(\omega \, \mu_\bal(\omega)\big) > 0$, and reciprocally. To conclude,  we compute (using that $\omega\in \bbR$ is not a magnetic or electric resonance):
		\begin{equation}\label{eq.positvity}
		\left\{	\; \begin{array}{lll} 
			\ds \operatorname{Im}  \big(\omega \, \varepsilon_\bal(\omega)\big)=\varepsilon_0 \, |\omega|^2\,  \sum_{j=1}^{N_e} \alpha_{e,j} \,  \frac{\Omega^2_{e,j} } {|q_{e,j}(\omega)|^2},\\ [18pt]
			\ds \operatorname{Im}  \big(\omega \, \mu_\bal(\omega)\big)=\mu_0 \, |\omega|^2\,  \sum_{\ell=1}^{N_m} \,\alpha_{m,\ell} \, \frac{\Omega^2_{m,\ell} } {|q_{m,\ell}(\omega)|^2},
			\end{array}  \right.
		\end{equation}
	which shows, since $\bal \neq 0$,  that $\operatorname{Im}  \big(\omega \, \varepsilon_\bal(\omega)\big)> 0$ or $\operatorname{Im}  \big(\omega \, \mu_\bal(\omega)\big) > 0$.	
	\end{proof}
\noindent Due to the strict inequality $\operatorname{Im}  \, \omega_{n}^\bal(|\bk|) < 0$, for $\bk \neq 0$, the amplitude of the corresponding plane wave decays exponentially when the time $t$ goes to infinity, see \eqref{planewave2}. However, the rate of decay degenerates,  that is to say $\operatorname{Im}  \, \omega_{n}^\bal(|\bk|) \rightarrow  0$, in the  following limit cases:
\begin{itemize} 
	\item [(i)] when $|\bf k| \rightarrow + \infty$ for $n =2 N+1, \, 2N+2$ or  when $p_n^\bal  \in \R$, 
	\item [(ii)] when $|{\bf k}| \rightarrow 0$ for $   z_n^\bal \in \R$.
	\end{itemize}
For treating the case $|\bf k| \rightarrow + \infty$, we introduce the set of indices
\begin{equation} \label{defNinfty}
	{\cal N}^\bal_\infty : =\, \Big\{ n \in \{1, \cdots, 2N\} \, / \, p_n^\bal \in \R\Big\},
	\end{equation}
	 so that  $\{p_n^\bal ,  n \in 	{\cal N}^\bal_\infty  \}$ is a rearrangement of the set or resonances $\text{R}_{\text e} \cup \text{R}_{\text m}$ (see definition \ref{resonances}). We shall also introduce the subset ${\cal N}^\bal_r$ of ${\cal N}^\bal_\infty$ defined by 
\begin{equation} \label{defNr} 
	{\cal N}^\bal_s: =\, \Big\{ n \in \{1, \cdots, 2N\} \, / \, p_n^\bal \in \text{R}_{\text e}^s \cup \text{R}_{\text m}^s \Big\},
\end{equation}
where $\text{R}_{\text e}^s$ and $\text{R}_{\text m}^s$ are defined in \eqref{defRe} and  \eqref{defRm}. 
\\ [12pt]
\noindent	For the sequel, we denote $c>0$ the speed of light characterized by $\varepsilon_0 \, \mu_0 \, c^2 = 1.$ The next three Lemmas \ref{thm_dispersion_asymptotics_1} to \ref{lem_asymptoticsHF1bis} provide the asymptotic expansions of $\operatorname{Im}  \, 	\omega^\bal_{n}(|\bk|)$ when $|\bk| \rightarrow + \infty$ and when it tends to 0, i. e. for $n \in {\cal N}^\bal_r \cup \{2N+1, 2N+2\}$. The proof of these lemmas is essentially computational and we refer the reader to \cite{cas-jol-ros-22-bis}, and more precisely to Lemmas 4.4, Lemma 4.7 and  its corollary  4.7 and Lemma 4.11, for the details.
\begin{Lem} \label{thm_dispersion_asymptotics_1} $[$ High frequency asymptotic  of $\operatorname{Im}  \, \omega^\bal_{n}(|\bk|) \,] $  for $n=2N+1, \, 2N+2$ and $|\bk|\to +\infty$,
	\begin{equation}\label{eq.asymptinf}
			\operatorname{Im}  \, \omega^\bal_{n}(|\bk|)=-  \, \frac{A_\infty}{2\, c^2|\bk|^2}\, + O(|\bk|^{-4}), \mbox{ with } A_{\infty}= \sum_{j=1}^{N_e} \alpha_{e,j}\, \Omega_{e,j}^2 + \sum_{\ell=1}^{N_m} \alpha_{m,\ell}\,\Omega_{m,\ell}^2 > 0.
		\end{equation}
		\end{Lem} 
\begin{Lem} \label{lem_asymptoticsHF1} $[$ High frequency asymptotics of $\operatorname{Im}  \, \omega_{n}^\bal(|\bk|), n \in {\cal N}^\bal_\infty$, generic case  $]$  
		If the medium is not weakly dissipative, for any $n \in {\cal N}^\bal_\infty$, for some $A_{n,2}^\infty > 0$ (see remark \ref{rem_value}) 
		\begin{equation}\label{eq.asymptps}
			\operatorname{Im}  \, 	\omega^\bal_{n}(|\bk|)= - \frac{A_{n,2}^\infty}{2\, c^2|\bk|^2}  + o(|\bk|^{-2}),\ \mbox{ as } |\bk|\to +\infty.
		\end{equation} 
\end{Lem}
\begin{Lem}  \label{lem_asymptoticsHF1bis} $[$  High frequency asymptotics  of $\operatorname{Im}  \, \omega_{n}(|\bk|), n \in {\cal N}^\bal_s$, weakly dissipative case $]$  \\
Suppose that the medium is electrically weakly dissipative and $n \in {\cal N}^\bal_s$. 
 If $p_n^\bal \in {\text R}_{\text e}^s$, which is not empty, see definition \ref{planewave2} (ii), then, for some $A_{n,4}^\infty> 0$,	\begin{equation}\label{eq.asymptpsbis}
\operatorname{Im}  \, 	\omega_{n}^\bal(|\bk|)= - \frac{A_{n,4}^\infty}{2\, c^4|\bk|^4}  + o(|\bk|^{-4}),\ \mbox{ as } |\bk|\to +\infty.
	\end{equation}
If $p_n^\bal \notin {\text R}_{\text e}^s$, the asymptotic expansion \eqref{eq.asymptps} still holds,
If the medium is magnetically weakly dissipative, the same results hold with ${\text R}_{\text m}$ instead of ${\text R}_{\text e}$.
\end{Lem}
\begin{Rem} \label{rem_value} $[$Value of $A_{n,2}^\infty$ $]$.  If the medium is not weakly dissipative, then  for each $n \in {\cal N}^\bal_\infty$, since $\{p_n^\bal ,  n \in 	{\cal N}^\bal_\infty \} = \text{R}_{\text e} \cup \text{R}_{\text m}$,
\begin{itemize}
\item Either $p_n= \pm \omega_{e,j}\in  {\text R}_{\text e}^s , \, \mbox{ and }  \displaystyle  A_{n,2}^\infty= \mbox{$\frac{1}{2}$} \, \varepsilon_0  \, \,\operatorname{Im}  \big(\omega_{e,j} \, \mu(\omega_{e,j})\big)  \,  \Omega_{e,j}^2 > 0$,
\item  Either  $p_n= \pm \omega_{m,\ell} \in {\text R}_{\text m}^s , \,  \mbox{ and }  \displaystyle  A_{n,2}^\infty= \mbox{$\frac{1}{2}$} \, \mu_0 \, \,\operatorname{Im}  \big(\omega_{\ell,m} \, \varepsilon(\omega_{m,\ell})\big)  \, \Omega_{m,\ell}^2 > 0,  $
\item	 Or $ p_n=\pm \omega_{e,j_0}=\pm \omega_{m,\ell_0} \in {\text R}_{\text e}^d={\text R}_{\text m}^d$ and  $$\displaystyle  A_{n,2}^\infty= p_n^2\,  \bigg( \frac{ \Omega_{m,\ell_0}^2}{2}  \sum_{j=1, j\neq j_0}^{N_e} \frac{\Omega_{e,j}^2 \alpha_{e,j} }{ |q_{e,j}(p_n)|^2}  + \frac{ \Omega_{e,j_0}^2}{2}   \sum_{\ell=1, \ell\neq \ell_0}^{N_m} \frac{\Omega_{m,\ell}^2 \alpha_{m,\ell} }{ |q_{m,\ell}(p_n)|^2} \bigg) > 0.  $$
\end{itemize}
\end{Rem}
\noindent
For treating the case $|\bf k| \rightarrow 0$, we introduce the set of indices
\begin{equation} \label{setsof incices} 
	{\cal N}^\bal_0: \, \Big\{ n \in \{1, \cdots, 2N+2\} \, / \, z_n^\bal \in \R\Big\},
\end{equation}
Note that, by the Assumption \ref{ireduciblebis},  ${\cal N}^\bal_0$ contains always at least one element since $0\in {\cal Z}_\bal \cap \bbR$ is real zero  of multiplicity two of $\mathcal{D}_{\alpha}$, and is reduced to a singleton   when the medium is electrically or magnetically non dissipative (cf. definition \ref{Resonances}). Indeed,  if the  medium is electrically  dissipative $ {\cal Z}_\bal \cap   \bbR=\{ 0\}  \cup \calZ_e^{\bal}$  (where we recall that $\calZ_e^{\bal}$ is the set of zeros of $ \varepsilon(\cdot)$) and all elements  of  $\calZ_e^{\bal}$ are simple zeros of $\mathcal{D}_{\alpha}$ (and in particular, they do not belong  to $\calZ_m^{\bal}$). The same properties hold for magnetically dissipative media with obvious changes.
\begin{Lem} [Low frequency asymptotics  of the dispersion relation]  \label{lem_asymptoticsHF2} 
	For $n \in {\cal N}^\bal_0$, 
	\begin{equation}\label{eq.asymptz0}
	\operatorname{Im}  \, 	\omega_{n}^\bal(|\bk|) = -A_{n,2}^0 \, c^2  \,  |\bk|^{2}    + o(|\bk|^{2}),\ \mbox{ as } |\bk|\to 0,
	\end{equation}
	where \begin{itemize}
	\item Either  $z_n=0$ and
	$ \ds
	A_{n,2}^0 =  \frac{\varepsilon_0^2\,  c^2}{2}\, \sum_{j=1}^{N_e} \frac{ \alpha_{e,j}\, \Omega_{e,j}^2 }{\omega_{e,j}^4} + \frac{\mu_0 \, c^2}{2}\, \sum_{\ell=1}^{N_m} \frac{  \alpha_{m,\ell}\, \Omega_{m,\ell}^2}{\omega_{m,\ell}^4} >0,
	$
\item  Either  	$z_n\in {\cal N}^\bal_0$ and $z_n\in \calZ_e^{\bal}$ (i.e. it is a real zero of $ \varepsilon(\cdot)$) and $\displaystyle  A_{n,2}^0 =\frac{(\omega \varepsilon)'(z)^{-1} }{z \mu(z)}>0$,
\item  Either  	$z_n\in {\cal N}^\bal_0$ and $z_n\in \calZ_m^{\bal}$ (i.e. it is a real zero of $ \mu(\cdot)$) and $\displaystyle A_{n,2}^0 =\frac{(\omega \mu)'(z)^{-1} }{z \varepsilon(z)}>0$.
\end{itemize}
\end{Lem}
\subsection{Long time behaviour of solutions: the result} \label{sec_decay} 
With the help of the Fourier transform in space $${\cal F}_x : u({\bf x}) \in L^2(\R^3) \rightarrow \widehat u({\bf k}) \in L^2(\R^3),$$
we can define the functional space  
	\begin{equation} \label{delLcalp}
		{\cal  L}_p (\R^3) =  \big \{ v \in \mathcal{S}'(\R^3) \; /  \; |\bk|^{-p} (1+|\bk|^{p}) \; \widehat v \in L^\infty(\R^3) \big \}
	\end{equation} 
which is a Banach space when  equipped  with the norm
	\begin{equation} \label{normcalp}
		\|v\|_{{\cal  L}_p} := \big\| \,  |\bk|^{-p} (1+|\bk|^{p}) \; \widehat v  \, \big\|_{L^\infty(\R^3)} .
	\end{equation}

\begin{Thm} \label{thm_Lorentz} Under the Assumptions \ref{ireducible}, \ref{ireduciblebis} and \ref{ireducible2}. The electromagnetic energy ${\cal E}(t)$ associated to the evolution system \eqref{DissipativeLorentzsystem} with intial conditions  \eqref{CI} tends to $0$ when $t \rightarrow +\infty$. Moreover, under the additional assumption
	$$
	(\bE_0, \bH_0)  \in H^s(\R^3)^3 \cap   {\cal L}_{p}(\R^3)^3 \times H^s(\R^3)^3 \cap   {\cal L}_{p} (\R^3)^3,
	$$
one has the  follwing decay estimates: 
\begin{itemize}
\item if \eqref{DissipativeLorentzlaws} is not weakly dissipative (the generic case), then
		\begin{equation} \label{polynomial_decayncr}
		{\cal E}(t) \leq C_s \; \frac{\|\bE_0 \|_{H^s}^2 + \|\bH_0 \|_{H^s}^2}{t^m} + C_p \,  \frac{ \|\bE_0 \|_{ \boldsymbol{\cal L}_{p}^N}^2 +  \|\bH_0 \|_{ \boldsymbol{\cal L}_{p}^N}^2}{t^{p+\frac{3}{2}}},  \quad  \forall \; t>0,
		\end{equation}
\item  if \eqref{DissipativeLorentzlaws} is weakly dissipative, then 
		\begin{equation} \label{polynomial_decaycr}
		{\cal E}(t)  \leq C_s \; \frac{\|\bE_0 \|_{H^s}^2 + \|\bH_0 \|_{H^s}^2}{t^\frac{s}{2}} + C_p \,  \frac{ \|\bE_0 \|_{ \boldsymbol{\cal L}_{p}^N}^2 +  \|\bH_0 \|_{ \boldsymbol{\cal L}_{p}^N}^2}{t^{p+\frac{3}{2}}},  \quad  \forall \; t>0.
		\end{equation}	
		\end{itemize}
\end{Thm}
\begin{Rem} [Comments on Theorem \ref{thm_Lorentz}]. \label{rem_comments}
\begin{itemize}
		\item In the theory of dynamical systems, estimates in inverse powers of $t$ such as \eqref{polynomial_decayncr} or \eqref{polynomial_decaycr} are called {\it polynimial stability}  by opposition to {\it exponential stability} when inverse powers of $t$  are replaced by decreasing exponentials. 
		\item This theorem corresponds to the Theorem 1.10 of the article \cite{cas-jol-ros-22-bis}. Similar results were obtained in  \cite{cas-jol-ros-22} by the same authors, using Lyapunov function methods. The convergence  ${\cal E}(t) \rightarrow 0$ can be seen in some particular cases as an application of an abstract result due to \cite{Figotin}, while polynomial stability results similar to \eqref{polynomial_decayncr} or \eqref{polynomial_decaycr} have been proven in \cite{Nicaise} under different assumptions that those of this chapter, using semi-group theory (see \cite{cas-jol-ros-22}, section 1.2, for more detailed bibliographical comments). 
	\item The decay estimates  \eqref{polynomial_decayncr} or \eqref{polynomial_decaycr} are sharp in the sense that similar lower bounds can be proven for adequate initial data, see  \cite{cas-jol-ros-22-bis}, theorem 1.14 and section 7.2.		
		
		\item In electromagnetism, a classical example of exponentially stable media is provided by conductive media, corresponding to the following local passive concstitutive laws
		$$
		\varepsilon(\omega) = \varepsilon_0 \, \Big( 1 + \frac{\sigma}{\rmi \, \omega} \Big), \quad \mu(\omega) = \mu_0,
		$$
		where $\sigma > 0$ is the conductivity of the medium. In that case, regardless the regularity of initial data,  the electromagnetic energy decays as $e^{-\sigma t}$. In comparison to Theorem \ref{thm_Lorentz},  this shows that dissipative Lorentz media are much less dissipative than conductive media. 

	\end{itemize}
\end{Rem}
	\begin{Rem} \label{remcalLp} $[$On the space ${\cal  L}_p (\R^3)$ $\!]$
	For $p \in \N$,  $	\|v\|_{{\cal  L}_p}$ is linked to the space 
	\begin{equation} \label{spacesL1weighted} 
		\displaystyle L^1_{p}(\R^3) := \big\{ v \in L^1(\R^3) \; / \; (1+|\bx|)^p \, v \in L^1(\R^3) \big\} \quad (\mbox{with } L^1_{0}(\R^3)= L^1(\R^3)),
	\end{equation}
	endowed with the norm
	$
	\|(1+|\bx|)^p \,v  \|_{L^1(\R^3)}.
	$
	Functions of $L^1_{p}(\R^3)$ have existing moments up to order $p$ which allows us to introduce (with $\alpha = (\alpha_1, \alpha_2, \alpha_3)\in \mathbb{N}^3$) 
	\begin{equation} \label{spacesL1weighted2} 
		\displaystyle L^1_{p,0}(\R^3) := \big\{ u \in L^1_{p}(\R^3) \; / \; \forall \; |\alpha| \leq p-1, \; \int \bx^\alpha \, u \, d{\bx} = 0  \big\}, \ \quad \mbox{(closed in $L^1_{p}(\R^3)$)}
	\end{equation}
	where $|\alpha| = \alpha_1+ \alpha_2+ \alpha_3$ and  $\bx^\alpha:= x_1^{\alpha_1}\, x_2^{\alpha_1}\, x_3^{\alpha_1}$. It can be proven (see \cite{cas-jol-ros-22-bis}, remark 1.7) that 
	\begin{equation} \label{inclusion} 
		\forall \; p \in \N, 	\quad \displaystyle L^1_{p,0}(\R^3) \subset {\cal  L}_p(\bbR^3)  \mbox{ with continuous injection}.
	\end{equation}
\end{Rem}	
\subsection{Long time behaviour of solutions : elements of proof }\label{sec_decay_proof}
{\bf Abstract form}. In the same framework as in section \ref{sec-abstract2}, we can rewrite 
\begin{equation}\label{eq.schro_diss}
	\frac{\md \, {\bf U}}{\md\, t} + \rmi\,\mathbb{A}_\bal  \, {\bf U}=0, \quad \mathbb{A}_\bal := \mathbb{A} - \rmi \, \mathbb{D}_\bal, 
\end{equation}
with unknown $ {\bf U} = \big( {\bf E}, {\bf H}, \bbP, \dot\bbP, \bbM, \dot\bbM\big) \in \boldsymbol{\cal H}$, where  $\mathbb{A}$ is the unbounded self-adjoint operator defined in section \ref{sec-abstract2}, see \eqref{eq.opAbis},  and $\mathbb{D}_\alpha$ the positive  bounded self-adjoint operator given by 
\begin{equation}\label{eq.opADalpha}
	\mathbb{D}_\bal := \, \begin{pmatrix}
		\quad 0 \quad &\quad 0 \quad & \quad 0 \quad& \quad 0 \quad & \quad 0 \quad & \quad 0 \quad \\[2pt]
	0 & 0 &0 & 0& 0&0 \\[2pt]
	0 & 0 &0  & 0 &0 &0\\[2pt]
0 & 0 & 0 & \mbox{diag} \; \alpha_{e,j} &0 &0\\[2pt]
	0 & 0 &0  & 0 &0 & 0 \\[2pt]
	0 &0& 0 & 0 & 0  & \mbox{diag} \; \alpha_{m,\ell}
\end{pmatrix}
\end{equation}
On its domain  $D(\bbA_{\bal})=D(\bbA)$,   $\mathbb{A}_\bal$  is  a  maximal dissipative operator. Thus, it generates a contraction semi-group of class $\mathcal{C}^0$ (see e.g. Theorem 4.3 page 14 of  \cite{Paz-83}). Hence, a Cauchy problem  associated to the evolution equation \eqref{eq.schro_diss} is well-posed and stable (see e.g. Theorem 1.3 page 102 of  \cite{Paz-83})).\\[6pt]
To deal with divergence free initial conditions, we can restrict ourselves to the subspace $\boldsymbol{\cal H}^\perp$ of $\boldsymbol{\cal H}$ made of fields $\big( {\bf E}, {\bf H}, \bbP, \dot\bbP, \bbM, \dot\bbM\big)$ that are all divergence free. This space (and its orthogonal) are preserved by $\mathbb{A}_\bal $. In other words, it is a  reducing space for the operator of $\mathbb{A}_\bal $. One denotes by $\mathbb{A}_\bal^\perp$
the restriction of $\mathbb{A}_\bal $ to $\boldsymbol{\cal H}^\perp$. We are thus interested in
\begin{equation}\label{sol_eq.schro_diss}
{\bf U}(t) = e^{- \rmi\, t \, \mathbb{A}_\bal^\perp} \, \bU_0.
\end{equation}
{\bf Fourier transform in space.} Applying the Fourier transform in space ${\cal F}_x $  to the equation \eqref{eq.schro_diss}, we see that, if $\widehat {\bf U}(\cdot,t) = {\cal F}_x  \big[{\bf U}(\cdot,t)\big]$, for each $\bk$, $\widehat {\bf U}(\bk,t)  \in {\bf C}^{2N+2},$ where ${\bf C} := \C^3$, satisfies the ordinary differential equation
\begin{equation}\label{eq.schr_diss_Fourier}
	\frac{\md \, \widehat{\bf U}}{\md\, t} (\bk,t) + \rmi\,\mathbb{A}_{\bal,\bk} \, {\bf U} (\bk,t)=0, \quad \mathbb{A}_{\bal,\bk} := \mathbb{A}_{\bf k} + \rmi \, \mathbb{D}_\bal, 
\end{equation}
where $\mathbb{A}_{\bal,\bk}$ is the hermitian matrix given in block form by 
\begin{equation}\label{eq.opAk}
	\mathbb{A}_{\bal,\bk} := \ \rmi\, \begin{pmatrix}
		0 &\varepsilon_0^{-1}\, \bk \, \times& 0&-{\mathbb M}_e & 0 & 0\\[2pt]
	- \mu_0^{-1}\, \bk \, \times & 0 &0 & 0& 0& -{\mathbb M}_m \\[2pt]
	0 & 0 &0  & 1&0 &0\\[2pt]
	{\mathbb M}_e^* & 0 &-\,\mbox{diag} \; \omega_{e,j}^2 &0 &0 &0\\[2pt]
	0 & 0 &0  & 0 &0 & 1\\[2pt]
	0 &{\mathbb M}_m^*& 0 & 0 & -\mbox{diag} \; \omega_{m,\ell}^2  & 0
\end{pmatrix},
\end{equation}
where the definition of the matrices ${\mathbb M}_e$ and ${\mathbb M}_m$ is a trivial adaptation of the one of the operators  ${\mathbb M}_e$ and ${\mathbb M}_m$ in section \ref{sec-abstract2}. We point out that for each $\bk\in \bbR^3 $, the linear reduced opererator $\mathbb{A}_{\bal,\bk}$, defined on the finite dimensional space  ${\bf C}^{2N+2}$, is dissipative.  \\ [12pt]
The divergence free condition ${\bf U}_0 \in  \boldsymbol{\cal H}^\perp$ can be traduced if Fourier by 
$$ 
\widehat{\bf U}_0(\bk) \in {\bf C}^{2N+2}_{\bk,\perp}, \quad {\bf C}_{\bk,\perp} \! :=  \!\big\{ {\bf V}  \in {\bf C} \; / \; \bk \cdot {\bf V} = 0 \big\} \quad (\mbox{with $\bk \cdot {\bf V}$ the inner product in ${\bf C}$}).
$$
The space $\C^{2N+2}_{\bk,\perp}$ is a reducing space of $\mathbb{A}_{\bal,\bk}$, we thus have
\begin{equation}\label{sol_eq.schro_diss-bis}
	\widehat{\bf U}(\bk,t) = e^{- \rmi\, t \, \mathbb{A}_{\bal,\bk}^\perp} \, \bU_0 (\bk), \qquad\qquad \mathbb{A}^\perp_{\bal,\bk} =  \mathbb{A}_{\bal,\bk} \big|_{{\bf C}^{2N+2}_{\bk,\perp}}.
\end{equation}
The long time behaviour of $\bU(t)$, via $\widehat \bU(\bk,t)$, is thus governed by the above exponentials, for which the spectrum of $\mathbb{A}_{\bal,\bk}^\perp$ obviously plays a fundamental role. \\ [12pt]
It is worthwhile emphasizing that the matrices  $\mathbb{A}_{\bal,\bk}$ (thus $\mathbb{A}_{\bal,\bk}^\perp$) are not normal, which is a source of technical difficulty in the analysis. 
~\\ [12pt]
{\bf Dispersion relation and spectra of  $\mathbb{A}_{\bal,\bk}^\perp$ and $\mathbb{A}_{\bal}^\perp$}. \\ [12pt]
The (immediate) link between the dispersion section \ref{sec_DispersionDissipative} and the operator $ \mathbb{A}_{\bal,\bk}^\perp$ (for $\bk\neq 0$) is
\begin{equation*} \label{spectrumAk} 
\omega \in \sigma\big( \mathbb{A}_{\bal,\bk}^\perp) \quad \Longleftrightarrow \quad {\cal D}_{\bal}(\omega) = |\bk|^2, 
\end{equation*}
thus, according to Theorem \ref{thm_dispersion_diss}, 
\begin{equation} \label{spectrum} 
	\sigma\big( \mathbb{A}_{\bal,\bk}^\perp) = \big \{ \omega_{n}^\bal( |\bk|),  \ 1 \leq n \leq 2N+2\big\}
\end{equation} 
and $\sigma\big( \mathbb{A}_{\bal,\bk}^\perp)$ is invariant by the transformation $\omega \to -\overline{\omega}$, i.e.  
$$ \omega_{n}^\bal( |\bk|)\in \sigma\big( \mathbb{A}_{\bal,\bk}^\perp) \Longleftrightarrow-\overline{\omega^\bal_{n}( |\bk|)} \in \sigma\big( \mathbb{A}_{\bal,\bk}^\perp).$$
\noindent If $\widehat {\bf U}:= {\cal F}_x {\bf U}$, $ \big[{\cal F}_x \big(\mathbb{A}_{\bal}^\perp {\bf U}\big)\big](\bk) = \mathbb{A}_{\bal,\bk}^\perp \big[\widehat {\bf U}(\bk)\big]$ for all $\bk$, i.e.  $ \ds \mathbb{A}_{\bal}^\perp = \int^\oplus  \mathbb{A}_{\bal,\bk}^\perp$. As a consequence, one shows by properties of Direct integral of dissipative operators (see e.g. corollary 3.3 of \cite{Ng-20}) that $\sigma\big(\mathbb{A}_{\bal}^\perp \big)$ is here given in terms of the eigenvalues  of the reduced operators  $\bbA_{\bk}^{\perp}$:
\begin{equation} 
\sigma\big(\mathbb{A}_{\bal}^\perp \big) = \bigcup_{n=1}^{2N+2}{\cal S}_n^\bal , \quad {\cal S}_n^\bal := \mbox{closure } \big\{ \omega^\bal_{n}\big( |\bk| \big), \bk \in \R ^3\}.
	\end{equation} 
Thus  $\sigma\big( \mathbb{A}_{\bal}^\perp)$ is also invariant by the transformation $\omega \to -\overline{\omega}$.
	Each $ {\cal S}_n^\bal$ is a curve arc in the complex plane half-plane $\bbC^-$ that joins $z_n^\bal$ and  $p_n^\bal$ (with $p_{2N+1}^\bal =-\infty$  and $p_{2N+2}^\bal= + \infty$ by choice  of  the indexing).
	In figure \ref{Fig_Spectrumdiss},  we plot the spectrum of  $\mathbb{A}_{\bal}^\perp$ for $N =3$ , when the medium is not electrically or magnetically non dissipative, with 2 real resonances. 
	\begin{figure}[h!] 
		\centerline{
		 \includegraphics[width=1 \textwidth]{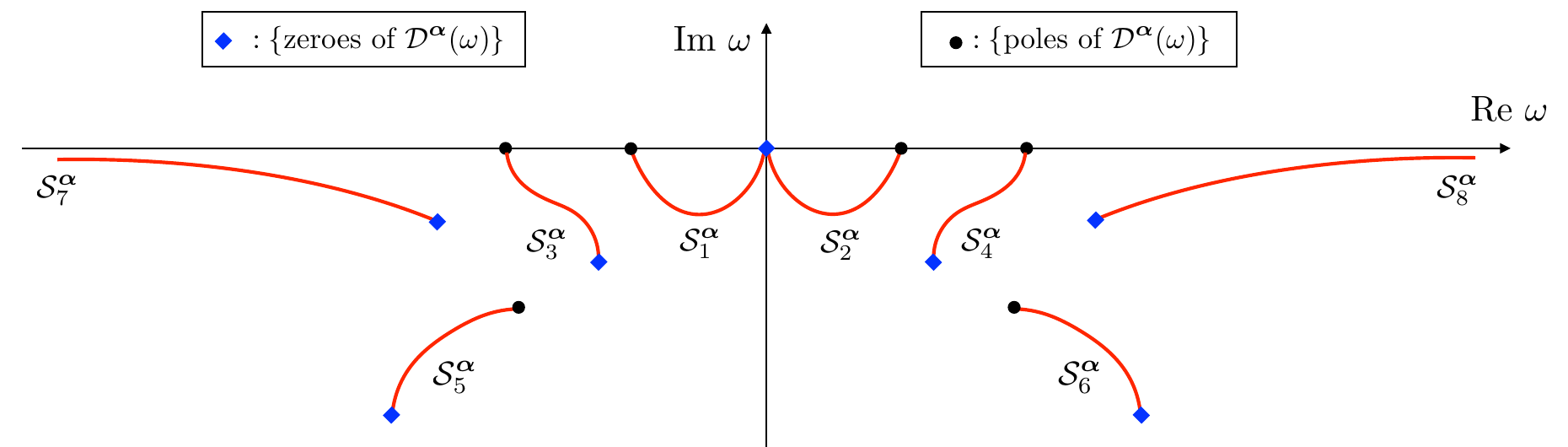}}
		\caption{Plot of the spectrum of  $\mathbb{A}_{\bal}^\perp$. The arcs ${\cal S}_n^\alpha$ play the same role as the (real) segments ${\cal S}_n$ in figure \ref{Fig_Disp}.
		}
		\label{Fig_Spectrumdiss}
	\end{figure}
\begin{Rem}
General properties on the structure of the spectrum of Maxwell's operators involving dissipative dispersive  materials (such its  localization  via numerical range techniques, the structure and classification  of the different types of essential spectra or the notion of spectral pollution with domain truncation, ...) have been analysed also  in non-homogeneous media : for transmission problems between two half-spaces in \cite{Bro-25}, in  locally perturbed media with cylindrical ends \cite{Fer-24} or in  more general locally perturbed media in \cite{Bog-23}. The approaches  used in these works are mainly based on spectral theory of pencils of operators.
\end{Rem}
\noindent 	{\bf Towards the estimate of ${\cal E}(t)$ via Plancherel's theorem}. \\ [12pt]
In the following,  to compare two positive functions $f(y)$ and $g(y)$,  where $y \in Y$ and $y = \bx, \bk,t$ or any combination of the variables, we introduce the notation:
		\begin{equation*} \label{notation2}
			f \lesssim g  \ \Longleftrightarrow \ \exists\;  C > 0  \quad  \mid \  f(y) \leq C \;  g(y), \quad  \forall \; y \in Y,
		\end{equation*} 
		where the constant  $C$ does not depend on $y$. 
\\ [12pt]
	According to the properties of the functions $\omega_{n}^\bal( |\bk|)$, in particular \eqref{limits_omega_n_alpha}, it is natural to split the space of dual variables $\bk$ into three regions: $\R^3 = {\bf LF} \cup {\bf MF} \cup {\bf HF}, \ $ where
	$$
	 \quad {\bf LF}  :=\big \{ 0<|\bk| < M \big\},  \quad {\bf MF}  := \big\{ m \leq |\bk| \leq  M \big\},  \ \mbox{ and } \ {\bf HF}  := \big\{ |\bk| > M \big\},
	$$
	with well chosen $0 < m < M < +\infty$. By Plancherel's theorem, one has 
	\begin{equation} \label{estiPlancherel}
		{\cal E}(t) \leq \|\bU(t)\|_{\cal H}^2 \; = \;  \int _{\bf LF} \big|	\widehat{\bf U}(\bk,t)\big|^2 \rmd\,  \bk +  \int _{\bf MF} \big|	\widehat{\bf U}(\bk,t)\big|^2  \rmd\,  \bk +  \int _{\bf HF} \big|	\widehat{\bf U}(\bk,t)\big|^2  \rmd\,  \bk . 
	\end{equation} 
According to  \eqref{limits_omega_n_alpha}, when $\bk$ describes ${\bf MF}$ (the middle space frequencies) the spectrum remains bounded away from the real axis, uniformly in $\bk$. As a consequence, there exists $\nu > 0$ (depending on $(m,M)$) such that 
$$
  \int _{\bf MF} \big|	\widehat{\bf U}(\bk,t)\big|^2  \rmd\,  \bk \lesssim e^{-\nu t},
$$
meaning that this term does not contribute to the right hand side of \eqref{polynomial_decayncr} (or  \eqref{polynomial_decaycr}), see Theorem 6.2 of \cite{cas-jol-ros-22-bis} for more details. In fact these two terms are explained by the integrals over ${\bf HF}$ and ${\bf LF}$ respectively. 
\\ [12pt]
We do not detail here the analysis of the first term in \eqref{estiPlancherel}, that concerns ${\bf LF}$ (the low space frequencies). This relies on  Lemma \ref{lem_asymptoticsHF2} and we refer the reader to  \cite{cas-jol-ros-22-bis}, section 7.1. One obtains
\begin{equation} \label{estiBF}
\int _{\bf BF} \big|	\widehat{\bf U}(\bk,t)\big|^2 \rmd\,  \bk \leq  C_p \,  \frac{ \|\bE_0 \|_{ \boldsymbol{\cal L}_{p}^N}^2 +  \|\bH_0 \|_{ \boldsymbol{\cal L}_{p}^N}^2}{t^{p+\frac{3}{2}}},
\end{equation}
which corresponds to the second part of  the estimates \eqref{polynomial_decayncr} and \eqref{polynomial_decaycr}. \\ [12pt]
In the sequel, we explain how to get the first part of  the estimates \eqref{polynomial_decayncr} and \eqref{polynomial_decaycr}, related  to the integral over ${\bf HF}$ (the high space frequencies), using  Lemmas \ref{thm_dispersion_asymptotics_1} and \ref{lem_asymptoticsHF1} (or Lemma \ref{lem_asymptoticsHF1bis}). \\ [12pt]
{\bf Spectral decomposition and estimate of ${\bf U}(\bk,t)$ for large $|\bk|$}. It can be shown, see \cite{cas-jol-ros-22-bis}, corollary 3.4,  that, at least for $M$ large enough, the matrices $\mathbb{A}_{\bal,\bk}^\perp$ are diagonalisable.  \\ [12pt] Let $\Pi_n^\bal(\bk)$ be the spectral projector of  $\mathbb{A}_{\bal,\bk}^\perp$  associated to the eigenvalue $ \omega_n^\bal(|\bk|)$, we have
\begin{equation}\label{eq.decompdiag}
\forall \; \bk \in {\bf HF}, \quad 	\widehat{\bf U}(\bk,t) =  \sum_{n=1}^{2N+2} \; e^{- \rmi \, \omega_n^{\bal}(|\bk|)\, t} \; \Pi_n^\bal(\bk) 	\widehat{\bf U}_0(\bk). 
\end{equation}
Using the Riesz-Dunford  functional calculus (see e.g. \cite{Dun-88} for a presentation of this holomorphic  functional calculus), one first  proves (see Lemma 4.13 of \cite{cas-jol-ros-22-bis}) that the part of the solution $\widehat{\bf U}(\bk,t)$ associated  to eigenvalues $ \omega_n^{\bal}(|\bk|)$ which tend to  poles $p_n^{\bal}\in \bbC^-$ decays exponentially (uniformly in $|\bk|$ for $\bk \in {\bf HF}$). 
 In other words, there exists $\delta>0$ such that:
\begin{equation} \label{eq.decayuniform}
\forall \; \bk \in {\bf HF}, \quad 	\Big|  \sum_{ n \notin {\cal N}^\bal_\infty  } \; e^{- \rmi \, \omega_n^{\bal}(|\bk|)\, t} \; \Pi_n^\bal(\bk) 	\widehat{\bf U}_0(\bk) \Big| \lesssim \rme^{-\delta \,t} |\widehat{\bf U}_0(\bk)|. 
\end{equation}
As $\mathbb{A}_{\bal,\bk}^\perp$ is not normal,  the projectors $ \Pi_n^\bal(\bk)  \in {\cal L}\big( \C^{2N+2}_{\bk,\perp}\big)$ involved  in the spectral decomposition \eqref{eq.decompdiag}   are not orthogonal. Thus, one needs to estimate the norm of  $\Pi_n^\bal(\bk) $ (with respect to $|\bk|$) when $n \in {\cal N}^\bal_\infty$. Indeed, it can be shown that these projectors  are uniformly bounded for $\bk \in {\bf HF}$ (this is proved in \cite{cas-jol-ros-22-bis}, see Lemmas 4.3, 4.6 and 4.10, using the Cauchy integral representation of these  spectral projectors \cite{Dun-88,kato}):
\begin{equation} \label{borne_projecteurs}
	\forall n \in {\cal N}^\bal_\infty, \quad \forall \; \bk \in {\bf HF}, \quad \big| \Pi_n^{\pm}(\bk)\big|_{{\cal L}({\bf C}^N)} \lesssim 1. 
	\end{equation}
	As a consequence , as  $ |e^{- \rmi \, \omega_n(|\bk|)\, t}|= e^{\, \operatorname{Im}   \omega_n(|\bk|)\, t}$, we deduce from \eqref{eq.decompdiag}, \eqref{eq.decayuniform} and  \eqref{borne_projecteurs}  that 
	\begin{equation} \label{borne_solfreq}
\forall \; \bk \in {\bf HF}, \quad  |	\widehat{\bf U}(\bk,t)|^2 \lesssim  \rme^{-2 \, \delta \,t} |\widehat{\bf U}_0(\bk)|^2+  \sum_{n\in {\cal N}^\bal_\infty \cup \{n \geq 2N+1 \}} \; e^{\,2 \, \operatorname{Im}    \omega_n^\bal(|\bk|)\, t} \; |	\widehat{\bf U}_0(\bk)|^2.
	\end{equation}
In the {\it non wealkly dissipative} case, for $n \in \{2N+1, 2N+2\} \cup {\cal N}^\bal_\infty$, according to the asymptotics of $\operatorname{Im}   \, \omega_n(|\bk|)$ which are all in $O(|\bk|^{-2})$ by Lemmas \ref{thm_dispersion_asymptotics_1} and \ref{lem_asymptoticsHF1}, one can find $\sigma > 0$ such that
\begin{equation} \label{lowerbound_imaginary part}
	\forall \;  n \in  \{2N+1, 2N+2\} \cup {\cal N}^\bal_\infty, \quad \forall \; \bk \in {\bf HF}, \quad \operatorname{Im} \,  \omega_n(|\bk|) \leq - \, \frac{{\sigma}}{2} \, |\bk|^{-2}.
\end{equation}
{\bf Estimate on the first term of  \eqref{estiPlancherel} (about large space frequencies in ${\bf HF}$).} From \eqref{lowerbound_imaginary part}, for  $\bk \in {\bf HF}$ and $n \in \{2N+1, 2N+2\} \cup {\cal N}^\bal_\infty$, one has  $e^{\,2 \,\operatorname{Im}   \omega_n(|\bk|)\, t} \leq e^{- \sigma \,  |\bk|^{-2} \, t}$ and therefore by virtue of \eqref{borne_solfreq}, one gets (for $M$ large enough):
$$ 
\int _{\bf HF} \big|	\widehat{\bf U}(\bk,t)\big|^2 \lesssim \int _{\bf HF}  e^{- \sigma \, {|\bk|^{-2} \, t}} \, |\widehat{\bf U}_0(\bk)|^2. 
$$
Thus, one deduces,  using the definition of Sobolev norms via Fourier transform, that
\begin{equation} \label{estimate}
\left| \begin{array}{lll} 
\ds \int _{\bf HF} \big|	\widehat{\bf U}(\bk,t)\big|^2 & \lesssim  & \ds t^{-m} \int _{\bf HF}  \Big( e^{- \, \sigma  \tfrac{t}{|\bk|^2}} \, \big( \tfrac{t}{|\bk|^2}\big)^m \Big) \;  |\bk|^{2m}  |\widehat{\bf U}_0(\bk)|^2 \\ [19pt]
 & \lesssim  & \ds C_m \; \frac{\| \bU_0\|_{H^m}^2}{t^m}, \quad  \mbox{with }  C_m := \sup_{\rho > 0} \, \big(\rho^{m} \, e^{-\sigma \, \rho}\big).
\end{array} \right.
\end{equation}
that is to say the first part of \eqref{polynomial_decayncr}.\\ [12pt]
In the {\it wealkly dissipative} case, the inequality \eqref{lowerbound_imaginary part} is no longer true for indices $n$ in the set ${\cal N}^\bal_s$ defined in \eqref{defNr}.  Instead, thanks to  Lemma \ref{lem_asymptoticsHF1bis}, we have, for some $\sigma > 0$:
\begin{equation} \label{lowerbound_imaginary part2}
	\forall \;  n \in   {\cal N}^\bal_s, \quad \forall \; \bk \in {\bf HF}, \quad  \operatorname{Im}  \,  \omega_n(|\bk|) \leq - \frac{\sigma}{2} \, |\bk|^{-4}.
\end{equation}
The indices $n$ in ${\cal N}^\bal_s$ are the ones that constrain  the energy decay : adapting computations of \eqref{estimate} (with $|\bk|^{-4}$ instead of $|\bk|^{-2}$) lead to the first term in \eqref{polynomial_decaycr} (with $s/2$ instead of $s$). 
\subsection*{Acknowledgements}
The authors would like to thank their colleague Maryna Kachanovska for her collaboration on the  article  \cite{cas-kach-jol-17} whose topic corresponds to the subjects treated from section \ref{sec-Laplace} to  section \ref{sec-Gen-Lorentz}.



\end{document}